\documentclass[12pt]{article}
\usepackage{amsfonts}
\usepackage[top=1in,bottom=1in,left=1in,right=1in]{geometry}
\usepackage{amssymb}
\usepackage{bm}
\usepackage{booktabs}
\usepackage{physics}
\usepackage[onehalfspacing]{setspace}
\usepackage[round]{natbib}
\usepackage{tablefootnote}
\usepackage{pdflscape}
\usepackage{geometry}
\usepackage{hyperref}
\usepackage{indentfirst}
\usepackage[flushleft]{threeparttable}
\usepackage{graphicx}
\usepackage{pgfplots}

\usepackage{tikz}
\usetikzlibrary{decorations.pathreplacing}
\newtheorem{theorem}{Theorem}

\newtheorem{satheorem}{Theorem A.}

\newtheorem{assumption}{Assumption}

\newtheorem{corollary}{Corollary}

\newtheorem{lemma}{Lemma}

\newtheorem{remark}{Remark}

\newtheorem{saassumption}{Assumption}
\newcommand{\indep}{\perp \!\!\! \perp}
\newenvironment{proof}[1][Proof]{\noindent\textbf{#1.} }{\ \rule{0.5em}{0.5em}}
\hypersetup{pdfborder = {0 0 0},colorlinks=true,linkcolor=black,urlcolor=blue,citecolor=black}

\newenvironment{customlegend}[1][]{%
	\begingroup
	\csname pgfplots@init@cleared@structures\endcsname
	\pgfplotsset{#1}%
}{%
	\csname pgfplots@createlegend\endcsname
	\endgroup
}%
\numberwithin{assumption}{section}
\numberwithin{saassumption}{section}
\def\addlegendimage{\csname pgfplots@addlegendimage\endcsname}
\begin{document}
	
	\title{Bounds for Treatment Effects in the Presence of Anticipatory Behavior}
	\author{Aibo Gong\thanks{School of Economics, Peking University}}	
	\maketitle	
	
	\begin{abstract}
	In program evaluations, units can often anticipate the implementation of a new policy before it occurs. Such anticipatory behavior can lead to units' outcomes becoming dependent on their future treatment assignments. In this paper, I employ a potential-outcomes framework to analyze the treatment effect with anticipation. I start with a classical difference-in-differences model with two time periods and provide identified sets with easy-to-implement estimation and inference strategies for causal parameters. Empirical applications and generalizations are provided. I illustrate my results by analyzing the effect of an early retirement incentive program for teachers, which the target units were likely to anticipate, on student achievement. The empirical results show the result can be overestimated by up to 30\% in the worst case and demonstrate the potential pitfalls of failing to consider anticipation in policy evaluation.
	\end{abstract}
	\thispagestyle{empty}
	\newpage
	\setcounter{page}{1}
	
\section{Introduction}
This paper accommodates the matter of anticipation in the analysis of treatment effects, by employing a potential-outcomes framework in a difference-in-differences (DID) model. The concept of anticipation is familiar to researchers in economics and social sciences, as seen in the work of \cite{malani2015interpreting}, as well as \cite{bovskovic2018much}, for example. When anticipation occurs, forward-looking units change their behavior in reaction to the possibility of a new policy, and thus, a treatment has an impact before its implementation. Therefore, considering the role of anticipation is crucial when evaluating an economic process and its outcome. However, despite its importance, most available published studies do not formally consider anticipation. People usually make a ``no anticipation'' assumption, combined with a procedure of dropping data closely before the treatment if this assumption is possibly violated, based on the argument that anticipation occurs only within a fixed time period prior to the introduction of a policy. Even in the few cases in which anticipation is taken into account, the anticipatory behavior is accounted for in a restricted manner, such as an ad-hoc restriction on units' forward-looking behavior like rational or adaptive expectations.

When anticipatory behavior takes place, the identification strategies commonly used with multiple periods, such as the DID model, fall apart. Consider an early retirement incentive (ERI) program for teachers nearing the age of retirement. If those teachers foresee the possibility of retiring early, their behavior might change before the program is introduced. Due to the effect of future treatment status on pre-treatment outcomes, the observable pre-treatment outcomes are no longer drawn from the distribution of potential outcomes, if the treatment never takes place. Individuals will change their responses according to how they expect to be treated in the future. Thus, further information about units' anticipatory behavior is required. However, such information is usually unattainable, because it is generally impossible to observe.

This paper provides novel strategies to build identified sets for treatment effects under assumptions restricting the anticipatory behavior. Easy-to-implement estimation and inference strategies are also provided. I start from a DID model with two time periods, and then generalize it to incorporate more complex models. I provide conditions for partial identification results of causal parameters when the anticipation status is unknown, and I incorporate anticipation in many widely used empirical designs. Employing a potential-outcomes framework, I analyze the treatment and the effects of anticipation based on the treatment rules, the anticipation assignments, and outcomes.

The departure from point identification starts with formulating restrictions on anticipatory behavior. In most cases, I do not have additional information, such as proxy variables, that helps us identify which participants have anticipated the policy change. As a result, I can say nothing about the pre-treatment distortion caused by anticipation. In this paper, I introduce a two-period DID model where anticipation occurs in the first period and the treatment occurs in the second period. Further, I introduce two natural assumptions to help construct bounds for the treatment effect in the absence of such additional information. The first is a bound for the proportion of anticipators within the treatment group. This bound should be available from observed data. It can be a constant, or a parameter that can be estimated. This practice is common in the literature, such as the work of \cite{manski2013deterrence}. The selection of this bound can vary from application to application, with one possible example being the treatment ratio. I provide models to motivate specific choices of the bound under various circumstances. The second assumption restricts the magnitude of the anticipatory effect. It requires that the absolute value of the anticipatory effect is no larger than that of the actual treatment effect. By doing so, I build a link between the magnitude of an anticipator's reaction and the response to the implementation of the policy. Based on this relationship, an inequality between the treatment effect and the average pre-treatment bias caused by anticipation can be constructed with the help of the proportion of anticipators discussed above. Therefore, I can find a corresponding treatment effect range for the anticipators by characterizing how they react and linking that anticipatory effect to the actual treatment effect. Under these two assumptions, the fraction of units that anticipate the policy change may vary, but the average distortion caused by anticipation is bounded and the parameter of interest is set identified. 

As for the implementation purpose, I propose estimation and inference strategies based on easy-to-implement modifications to existing methods. The identification strategy provides an identified set with perfectly correlated and proportional upper and lower bounds. I propose a uniformly valid confidence set for my estimators with some modifications to \cite{imbens2004confidence} under this specific setup. In their method, the upper and lower bounds of the confidence set are found by extending both sides of the identified set. The extended lengths are proportional to the standard errors of the bound estimators and differ between upper and lower bounds. However, suppose this method is applied directly here. In that case, it may run into a counterintuitive situation in which the confidence set for the treatment effect is shorter when the parameter is partially identified than when it is point-identified. I propose modifying this approach by extending both sides of the identified set by the same length proportional to the larger standard error of the two, which is a natural way to ensure the uniform validity. Analyzing this confidence set also provides researchers with further empirical implications. When the treatment and the anticipatory effect go in different directions, I find a specific range of t-statistics for the zero treatment effect null hypothesis. If the t-statistic obtained when anticipation is ignored falls within this range, the conclusion of whether to reject the null hypothesis does not change when considering anticipation. This confidence set also helps build a framework for sensitivity analysis on certain conclusions of interest by choosing different bounds for anticipation possibility. 

I apply the results of this paper to examine the effect of an early retirement incentive program on student achievement. This program is aimed at teachers near the age of retirement and offers them financial incentives to retire before becoming eligible for full pension benefits. If the program is anticipated, eligible teachers may react in advance of its introduction, and such behavior might affect students' grades. The empirical results illustrate the potential pitfalls of failing to consider anticipation in program evaluation: the effect can be greatly overestimated in the worst case. I also conduct a sensitivity check by analyzing the level of anticipation probability one is willing to tolerate while maintaining the consistency of the original conclusion. It shows the conclusion is robust even when about three fourths of target units anticipate. 

To permit the incorporation of anticipation in other common empirical setups, I provide several modifications. Instead of focusing only on the pre-treatment effect of anticipatory behavior in the treated group, I discuss the anticipatory behavior in the control group by introducing an imperfect anticipation setup, where individuals make mistakes while anticipating. Post-treatment effects of anticipatory behavior are also discussed. To be consistent with common empirical approaches, generalizations to include covariates, multiple periods, and nonlinear potential outcomes are provided and analyzed in the appendix.

This paper contributes to the literature on causal inference and program evaluation; see \cite{abadie2018econometric}, \cite{athey2017econometrics} for example. My paper is most closely related to the work of \cite{malani2015interpreting}, who discuss anticipatory behavior by interpreting the pre-trend phenomenon as a result of anticipation. The authors propose a parametric time series model in which anticipation is an expectation of the future treatment for everybody, by relying on the rational or adaptive expectation assumption. I incorporate the idea of anticipation in a DID framework with potential outcomes to remove parametric restrictions and allow heterogeneous anticipatory behavior among units. \cite{heckman2007dynamic}, under a different scenario, present a reduced form dynamic treatment effect model that also permits anticipation, but at the price of imposing further assumptions on the functional structure of the outcome equation. 

This paper also contributes to the literature on DID and event-study designs by considering anticipatory behavior. The additional anticipation could have an impact prior to the introduction of a policy. Therefore, the present research is related to the literature aiming at more robust inference and identification strategies that allow for non-parallel trends assumptions, and to papers focusing on pre-trend analysis. 

To interpret and deal with observed changes in outcomes prior to a treatment, \cite{manski2018right} propose a result on partial identification for the average treatment effect under ``bounded variation'' assumptions. These assumptions relax the parallel trends assumption by allowing for differences within a certain magnitude. \cite{rambachan2022more} follow the idea that pre-treatment differences in trends are informative about counterfactual post-treatment differences and provide identification and inference results based on several common restrictions of this relationship. \cite*{freyaldenhoven2019pre} propose a method that includes an additional covariate that is correlated with the outcomes through confounds only, and not treatments. \cite*{ye2021negative} propose a partial identification method for treatment effects with two groups of control units whose outcomes exhibit a negative correlation relative to the treated units. In this paper, I interpret pre-trends as a result of unobservable anticipation activities. People may change their behavior because of their anticipation of future treatment. If people have information and may benefit by acting on it before a treatment, anticipation is a reasonable explanation for an observed pre-treatment effect, even when the parallel trends assumption is valid.

This paper is also complementary to the causal interpretation of event-study coefficients; see \cite{borusyak2017revisiting}, \cite{sun2020estimating}, \cite{de2020two}, and \cite{goodman2021difference}. With a generalization to the longitudinal data, this paper can be regarded as relaxation of the ``no anticipation'' assumption of these papers. This paper is also related more generally to the partial identification literature. In the present study, partial identification is obtained through moment inequalities, a method that is discussed in \cite{molinari2020microeconometrics}.

The rest of this paper is organized as follows: Section 2 generalizes the commonly used DID model and introduces the basic setup about anticipation; Section 3 provides extra assumptions and shows readers how to build the identified sets; Section 4 describes estimation and inference; Section 5 provides an empirical application; Section 6 offers further discussion; and Section 7 concludes. The mathematical proofs, together with some additional results, discussions, and generalizations, are collected in the supplemental appendix.
\section{Setup and Assumptions}
To illustrate anticipation in program evaluation, I consider an early retirement incentive program available for teachers that offers experienced teachers financial incentives to retire before they would be eligible for full pension benefits. Suppose one is interested in the effect of this early retirement incentive program on students' grades. Anticipation from teachers, whether treated or not, can be expected for several reasons in this program. Teachers who anticipate may have received inside information from others, and they can also speculate based on changes that have already happened. Younger teachers ineligible for the program won't react to it regardless of the anticipation status in both cases. However, teachers who anticipate the program and decide to retire early may put in less effort than younger teachers. Such behavior may harm students' grades before the implementation of the early retirement incentive program, and ignoring anticipation can lead to a bias while analyzing the effect of this program. The fact that teachers can anticipate based on unobservable information and adjust their behavior accordingly to gain benefits implies the future treatment will have an effect before its adoption and distort the treatment effect estimation if ignored. An accurate assessment of anticipation is therefore essential for the program evaluation.
\subsection{The Basic Difference-in-Differences Model}
First I briefly describe the ``canonical'' two-period DID model in this section. As a well-understood starting point, this simple setting serves as a good baseline for understanding the approach I use.

Consider a model with two periods $t\in\{0,1\}$ and $n$ units, $i\in\{1,\dots,n\}$. Each unit is assigned an observable binary treatment $D_i$ that takes value $d\in\{0,1\}$ in the second period. The key identifying assumption requires that the treated and control group follow parallel trends in the absence of treatment, and the parameter of interest is the average treatment effect for treated (ATT).

Potential outcomes, defined below, depend on the time period and binary treatment status. The potential outcome for unit $i$ in period $t$ is denoted by the random variable $Y_{it}(d)$. Given a value of the implemented treatment $d$, the observed outcome of unit $i$ at period $t$, $Y_{it}$ can be written as
\begin{equation*}
	Y_{it}=\sum_{d\in\{0,1\}}Y_{it}(d)\mathbb{I}(D_i=d)=D_iY_{it}(1)+(1-D_i)Y_{it}(0),
\end{equation*}
and the parameter of interest $\mu=\mathbb{E}[Y_{i1}(1)-Y_{i1}(0)|D_i=1]$. For identifying purpose, I need to assume \textit{``Parallel Trends''} and \textit{``No Anticipation''}, which require
\begin{equation*}
	\mathbb{E}[Y_{i1}(0)-Y_{i0}(0)|D_i=1]=\mathbb{E}[Y_{i1}(0)-Y_{i0}(0)|D_i=0],
\end{equation*}
and
\begin{equation*}
	Y_{i0}(0)=Y_{i0}(1)\qquad\text{for all $i$.}
\end{equation*}
Under these two assumptions, the parameter of interest $\mu$ is identified, and for estimation purposes, I need to further put independent restrictions on the sampling process. The key idea here is that following the parallel trends assumption, one can use the change in the control group to mimic that in the treated group in the absence of treatment and get information about the unobservable potential outcomes for the treated group in the absence of treatment in the post-treatment period. However, as pointed out, this approach requires no anticipatory behavior, which assumes that the post-treatment status should have no impact on pre-treatment outcomes. This assumption may be too restrictive in some situations, for example, the early retirement incentive program mentioned above. Thus, figuring out a way to accommodate anticipatory behavior under this DID setup is important.
\subsection{Introducing Anticipatory Behavior}
To deal with the unobservable anticipatory behavior, I introduce another indicator for anticipation status. Suppose that in the first period, each unit has an unobservable binary anticipation status $A_i$ that takes value $a\in\{0,1\}$. Here, $a=1$ means this unit anticipates the future, and $a=0$ means this unit does not anticipate the future. The potential outcomes can now be written as $Y_{it}(a,d)$. Introducing a second index in the expression of potential outcomes is common when analyzing indirect effects, for example, analysis of spillover effects in \cite{vazquez2021identification}. By implicitly assuming perfect anticipation, which means the anticipated treatment status should be the same as the actual treatment, I focus only on the pre-treatment anticipatory behavior in the treated group at this time. Later, I discuss anticipatory behavior in the control group and the post-treatment effect of the anticipatory behavior. After introducing another index for the anticipatory behavior, potential outcomes, defined below, can now depend on the binary treatment and anticipation status. I refer to the existence of the latent treatment in the first period as \textit{anticipation} and the effect of the anticipatory behavior on unit $i$'s potential outcome before the treatment occurs as the \textit{anticipatory effect}.

As stated above, I focus only on the pre-treatment anticipatory behavior in the treated group now, which means I allow $Y_{i0}(0,1)$ and $Y_{i0}(1,1)$ to be different from each other with no changes on $Y_{i0}(0)$, $Y_{i1}(0)$ and $Y_{i1}(1)$. With a little abuse of the notation, I use the single index expression $Y_{it}(d)$ for the potential outcomes when the anticipation status makes no difference. Because of the existence of anticipatory behavior, the observed pre-treatment outcomes of the treated group are now a mixture of those who anticipate and those who don't. This mixture brings extra difficulties in identification, because one cannot distinguish the anticipators from those who do not anticipate, and further assumptions are needed. To start with, I consider those assumptions that come from the canonical DID model.
\begin{assumption}[Sampling]
	\label{sampling}
	$\{Y_{i0}(0,1), Y_{i0}(1,1), Y_{i0}(0), Y_{i1}(0), Y_{i1}(1), D_{i},A_{i}\}_{i=1}^{n}$  are independently and identically distributed across $i$.
\end{assumption}

Assumption \ref{sampling} models the sampling process and states the potential outcomes, treatments, and unobservable anticipation status to be independent and identically distributed across units so that expectations are not indexed by $i$.
\begin{assumption}
	\label{onlychannel}
	$Y_{i0}(0)=Y_{i0}(0,1).$
\end{assumption}

Recall that $Y_{i0}(0)$ represents the pre-treatment potential outcome if one will not get treated. This assumption requires future treatment does not make a difference for the pre-treatment outcome in the absence of anticipatory behavior. One can interpret this assumption in a way that anticipation of a future treatment is the only channel through which future events affect the present. In the early retirement incentive program example, this assumption implies there is no difference in the pre-treatment grades for students taught by the same teacher regardless of the teacher's decision about early retirement if he has no anticipation of the program.

The table below shows potential outcomes in the DID model with two periods.
\begin{center}
	\begin{tabular}{c | c c} 
		\hline
		& $t$=0 & $t$=1 \\
		\hline
		$D_i=0$ & $Y_{i0}(0)$ & $Y_{i1}(0)$  \\ 
		
		$A_i=0$, $D_i=1$ & $Y_{i0}(0,1)$ & $Y_{i1}(1)$  \\
		
		$A_i=1$, $D_i=1$ & $Y_{i0}(1,1)$ & $Y_{i1}(1)$  \\
		\hline
	\end{tabular}
\end{center}
Compared with the commonly used DID framework, the critical difference is that the observed pre-treatment outcome $Y_{i0}$ for the treated group is a mixture of the potential outcomes for those who do not anticipate $Y_{i0}(0,1)$ and those who do anticipate $Y_{i0}(1,1)$ in the treated group in the pre-treatment period. $\mathbb{E}[Y_{i0}|D_i=1]$ is no longer a good measure for the first-period potential outcome without treatment for the treated group. 

Under this setup, the parameter of interest I focus on is still the average treatment effect for treated (ATT) with a slight modification:
\[\mu_g=\mathbb{E}[g(Y_{i1}(1))-g(Y_{i1}(0))|D_i=1],\]
where $g(.)$ is a known measurable real function with $\mathbb{E}\lvert g(Y)\rvert<\infty$. Define the corresponding anticipatory effect for anticipators as \[\tau_g=\mathbb{E}[g(Y_{i0}(1,1))-g(Y_{i0}(0,1))|D_i=1,A_i=1].\]
The $g(.)$ function is slightly generalized from the commonly defined ATT. When $g(.)$ is the identity function, $\mu_g$ is the widely used ATT. If $g(.)$ is an indicator function such as $g_{u}(Y)=\mathbb{I}(Y\le u)$, then $\mu_g$ can be interpreted as the change in the probability of the outcomes being no more than a specific cutoff $u$ and can be used to help identify the distribution of potential outcomes. Different choices of this $g(.)$ function lead to different estimators. Introducing the $g$ function enables handling of some nonlinear structures for parameters I am interested in. For simplicity of notation, I write $\mu_g$ as $\mu$ and $\tau_g$ as $\tau$ when $g(.)$ is the identity function.
\begin{assumption}[Parallel Trends]
	\label{paralleltrends}
	\begin{eqnarray*} 
		\mathbb{E}[g(Y_{i1}(0))-g(Y_{i0}(0))|D_i=1]&=&\mathbb{E}[g(Y_{i1}(0))-g(Y_{i0}(0))|D_i=0].
	\end{eqnarray*}
\end{assumption}

Although the expression seems to be the same as the parallel trends assumption in the canonical DID model, Assumption \ref{paralleltrends} requires that the treatment and control group change following parallel trends before and after the treatment in the absence of both anticipation and the treatment.
\begin{remark}
	As explained above, for the benchmark model and basic results, I focus only on the pre-treatment anticipatory behavior in the treated group. For the anticipatory behavior in the control group, one can understand this simplification as an implicit assumption that all people anticipate their future correctly, so that if people anticipate no future treatment, they have no incentive to change their behavior. Whether one anticipates the future and realizes no treatment or no anticipation will not make any difference. Further modifications to incorporate anticipatory behavior in the control group and the post-treatment effect of anticipatory behavior will be discussed later.
\end{remark}

Then, I briefly discuss what the commonly used DID estimator estimates without considering anticipation and compare the finding with the parameter that I am interested in. Throughout this section I choose $g(.)$ to be the identity function.

In the DID regression model with two periods,
\[Y_{it}=\beta_0+\beta_1 t+\beta_2 D_i+\beta_3 tD_i+\varepsilon_{it}.\]
Under Assumptions \ref{sampling}-\ref{paralleltrends}, the coefficient of interest, $\beta_3$, can be written as
\begin{eqnarray*}
	\beta_3 &=&\mathbb{E}[Y_{i1}|D_i=1]-\mathbb{E}[Y_{i0}|D_i=1]-\mathbb{E}[Y_{i1}|D_i=0]+\mathbb{E}[Y_{i0}|D_i=0]\\	
	&=& \mu-\mathbb{P}[A_i=1|D_i=1]\tau\\
	&=& \mathbb{P}[A_i=1|D_i=1](\mu-\tau)+(1-\mathbb{P}[A_i=1|D_i=1])\mu.
\end{eqnarray*}

If the DID estimator is used directly, it will suffer from a bias equal to the average distortion caused by anticipation. This bias arises from the fact that the observable pre-treatment outcomes for the treated group do not reflect the potential outcomes for them without the treatment. Those who anticipate have already reacted in the first period and deviated from the parallel-trends benchmark. Thus, applying the DID estimator directly suffers from a bias determined by both the proportion of those who anticipate and the magnitude of anticipatory effects. The last equality points out that this parameter can also be written as a weighted average of the treatment effect $\mu$ for those who do not anticipate and the net treatment effect after the adoption of the policy for those who do anticipate $\mu-\tau$. In general, the relationship between this estimand and the treatment effect depends on the sign of the anticipatory effects. Suppose the treatment and anticipatory effects have the same sign. In that case, anticipation will drive the DID estimator toward zero relative to the treatment effect because of contamination. The graph below captures the idea of this distortion. Figure \ref{DID graph} shows the result obtained by applying the DID estimator directly, whereas Figure \ref{anticipation graph} describes the situation that considers anticipation.  Anticipation causes the distortion between $\beta_3$ and $\mu$.		
\begin{landscape}	
	\begin{figure}[bp]
		\centering
		\begin{minipage}[b]{0.7\textwidth}
			\centering
			\begin{tikzpicture}
				\draw[thick,->] (0,0) -- (7,0) node[anchor=north west] {x};
				\draw[thick,->] (0,0) -- (0,7) node[anchor=east] {y};
				\filldraw[black] (1,1.5) circle (2pt);
				\filldraw[black] (5,3) circle (2pt);
				\filldraw[black] (5,6.5) circle (2pt);
				\filldraw[black] (1,4) circle (2pt);
				\draw[gray, thick] (1,4) -- (5,6.5);
				\draw[gray, thick] (1,1.5) -- (5,3);
				\draw (1,1pt) -- (1,-1pt) node[anchor=north] {t=0};
				\draw (5,1pt) -- (5,-1pt) node[anchor=north] {t=1};
				\draw (1pt,1.5) -- (-1pt,1.5) node[anchor=east] {$\mathbb{E}[Y_{i0}|D_i=0]$};
				\draw (1pt,3) -- (-1pt,3) node[anchor=east] {$\mathbb{E}[Y_{i1}|D_i=0]$}; 
				\draw (1pt,4) -- (-1pt,4) node[anchor=east] {$\mathbb{E}[Y_{i0}|D_i=1]$};
				\draw (1pt,6.5) -- (-1pt,6.5) node[anchor=east] {$\mathbb{E}[Y_{i1}|D_i=1]$};
				\draw[dashed](0,6.5) -- (5,6.5);
				\draw[dashed](0,4) -- (1,4);
				\draw[dashed](0,3) -- (5,3);
				\draw[dashed](0,1.5) -- (1,1.5);
				
				\draw[black] (5,5.5) circle (2pt);
				\draw[dashed](1,4) -- (5,5.5);
				\draw[decorate,decoration={brace,amplitude=5pt,mirror,raise=4pt},yshift=0pt]
				(5,5.5) -- (5,6.5)node [black,midway,xshift=2cm]{\footnotesize
					DID estimator};
			\end{tikzpicture}		
			\caption{Difference-In-Differences without Anticipation}
			\label{DID graph}
		\end{minipage}
		\begin{minipage}[b]{0.7\textwidth}
			\centering
			\begin{tikzpicture}
				\draw[thick,->] (0,0) -- (7,0) node[anchor=north west] {x};
				\draw[thick,->] (0,0) -- (0,7) node[anchor=east] {y};
				\filldraw[black] (1,1.5) circle (2pt);
				\filldraw[black] (5,3) circle (2pt);
				\filldraw[black] (5,6.5) circle (2pt);
				\filldraw[black] (1,4) circle (2pt);
				\draw[gray, thick] (1,4) -- (5,6.5);
				\draw[gray, thick] (1,1.5) -- (5,3);
				\draw (1,1pt) -- (1,-1pt) node[anchor=north] {t=0};
				\draw (5,1pt) -- (5,-1pt) node[anchor=north] {t=1};
				\draw (1pt,1.5) -- (-1pt,1.5) node[anchor=east] {$\mathbb{E}[Y_{i0}|D_i=0]$};
				\draw (1pt,3) -- (-1pt,3) node[anchor=east] {$\mathbb{E}[Y_{i1}|D_i=0]$}; 
				\draw (1pt,4) -- (-1pt,4) node[anchor=east] {$\mathbb{E}[Y_{i0}|D_i=1]$};
				\draw (1pt,6.5) -- (-1pt,6.5) node[anchor=east] {$\mathbb{E}[Y_{i1}|D_i=1]$};
				\draw[dashed](0,6.5) -- (5,6.5);
				\draw[dashed](0,4) -- (1,4);
				\draw[dashed](0,3) -- (5,3);
				\draw[dashed](0,1.5) -- (1,1.5);
				
				\draw[black] (1,3.5) circle (2pt);
				\draw[black] (5,5) circle (2pt);
				\draw[dashed](1,3.5) -- (5,5);
				\draw(1pt,5) -- (-1pt,5) node[anchor=east] {$\mathbb{E}[Y_{i1}(0)|D_i=1]$};
				\draw (1pt,3.5) -- (-1pt,3.5) node[anchor=east] {$\mathbb{E}[Y_{i0}(0)|D_i=1]$};
				\draw[dashed](0,3.5) -- (1,3.5);
				\draw[decorate,decoration={brace,amplitude=5pt,mirror,raise=4pt},yshift=0pt]
				(5,5) -- (5,6.5) node [black,midway,xshift=2.7cm] {\footnotesize
					$\mu=\mathbb{E}[Y_{i1}(1)-Y_{i1}(0)|D_i=1]$};
				\draw [decorate,decoration={brace,amplitude=3pt,mirror,raise=4pt},yshift=0pt]
				(1,3.5) -- (1,4);
				\draw[-stealth](1.4,3.75) -- (1.65,3.5);
				\draw(1.65cm-1pt,3.5) -- (1.65cm+1pt,3.5) node[anchor=west] {$\mathbb{P}[A_i=1|D_i=1]\tau$};
				\draw[dashed](0,5) -- (5,5);	
			\end{tikzpicture}
			\caption{With Anticipation}
			\label{anticipation graph}
		\end{minipage}
	\end{figure}
\end{landscape}
\section{Upper and Lower Bounds for Treatment Effects}
\subsection{Main Results}
Anticipation makes the commonly used DID estimator a mixture of anticipatory and treatment effects. The fundamental difficulty in obtaining identification is distinguishing between those who anticipate and those who do not. This section introduces several assumptions to build upper and lower bounds for treatment effects under different circumstances. Motivations for specific assumptions are provided. The following results link observed outcomes, potential outcomes, treatment assignments, and anticipation status and are used in further discussions.

The analysis above shows that two unobservable variables are contributing to the pre-treatment distortion. One is the possibility of treated units anticipating, $\mathbb{P}[A_i=1|D_i=1]$, and the other is the anticipatory effect for anticipators $\tau_g$. These two variables both need to be analyzed to recover the treatment effect. If a reasonable proxy is available for the anticipation treatment, one can use this proxy variable to measure the anticipation status for each unit. However, such a proxy variable is not always available. To overcome the difficulty of not being able to distinguish people who anticipate from others, I introduce a bounding parameter $\pi\in(0,1)$ that summarizes how the anticipation probability can be bounded. Here, $\pi$ is a parameter that can be obtained from available information, including the treatment assignment and outcomes of units, and it can be either constant or at least estimated from observable terms. Researchers can choose a $\pi$ based on their empirical setups, and I discuss possible choices of $\pi$ in the next section.
\begin{assumption}
	\label{prob_bound}
	$\mathbb{P}[A_i=1|D_i=1]\le \pi$.
\end{assumption}

Although I cannot observe the anticipatory effect $\tau_g$ directly,  I can build a relationship between it and the treatment effect for treated $\mu_g$. $\tau_g$ is caused by people's anticipation of a possible future treatment and behavior before the treatment to gain benefit. People's reactions and behavior are guided by their own guesses of the future policy. On the other hand, $\mu_g$ measures the treatment effect treated people receive when the policy is adopted. This effect happens based on revealed policy and treatment status. For example, suppose someone is going to sell his property at a lower-than-usual price because of anticipation of a possible negative price shock in the future. In this case, he has no reason to accept a price that is even lower than the price when the shock comes. Therefore, one might reasonably expects that the magnitude of treatment effects should be no smaller than that of the anticipatory effect, because the former is a reaction based on known information, whereas the latter one is based on uncertainty. In the example of the early retirement incentive program, this statement requires that the effect caused by teachers who exert less effort because they anticipate a possible early-retirement opportunity is no larger than the treatment effect when the early retirement incentive program is implemented. Because the anticipatory effect and treatment effect may not be in the same direction, I impose restrictions only on the magnitude.
\begin{assumption}
	\label{magnitude}
	$\abs{\tau_g}\le\abs{\mu_g}$
\end{assumption}

Assumptions \ref{prob_bound} and \ref{magnitude} help build bounds for the two unobservable terms, the proportion of anticipators among treated and the magnitude of the anticipatory effect separately based on the above assumptions. The parameter of interest $\mu_g$ is partially identified using observed variables, especially with the help of the commonly used DID estimand.
\begin{theorem}
	Under Assumptions \ref{sampling}-\ref{paralleltrends} and \ref{prob_bound}-\ref{magnitude}, the parameter of interest $\mu_g$ is partially identified via a closed interval in the following form:
	\[\mu_g\in m_g\left[\min\left\{1,\frac{1}{1-\text{sgn}(\tau_g\mu_g)\pi}\right\},\max\left\{1,\frac{1}{1-\text{sgn}(\tau_g\mu_g)\pi}\right\}\right]\]
	with $m_g=\mathbb{E}[g(Y_{i1})-g(Y_{i0})|D_i=1]-\mathbb{E}[g(Y_{i1})-g(Y_{i0})|D_i=0]$.
	\label{chap2:benchmark}
\end{theorem}

Theorem \ref{chap2:benchmark} points out that the treatment effect is located in an interval where the DID parameter without anticipation is one of its bounds. The other bound is obtained by enlarging or reducing it by a specific ratio depending on the bounding parameter $\pi$ and signs of treatment and anticipation effects. As shown in Figures \ref{DID graph} and \ref{anticipation graph}, the distortion happens only within the group of treated and anticipate units, so once the sign of the anticipatory effect is determined, the sign of the bias is also determined. The DID parameter without anticipation is by design one side of the interval. If anticipatory behavior happens in both the control and the treated group, distortions happen in both groups, and the sign of the bias is ambiguous. The distortion bias has a limited magnitude restricted by both the bounding parameter $\pi$ and the treatment effect magnitude, so I can build partial identification results for the parameter of interest based on observables.

Although I impose the bounds for anticipation probability and magnitude restrictions, these assumptions are not the only way to build partial identification results for the parameter of interest with anticipation. Empirical setups may exist in which these assumptions are not reasonable, and researchers would like to impose alternative assumptions, such as bounded outcomes or further conditional independent restrictions. These assumptions are also reasonable under specific situations, such as when the $g(.)$ function I am interested in is bounded by itself. I do not argue that the identified set under one assumption is tighter than the other, so I should choose one of them; rather, different sets of assumptions may be reasonable under different empirical circumstances. Incorporating more combinations of alternative assumptions and providing identification results allow us to incorporate anticipation in more situations and give researchers the freedom to modify assumptions based on the empirical setup. I propose several different combinations of assumptions as well as corresponding upper and lower bounds expressions of the treatment effects in the appendix.
\subsection{Choice of $\pi$}
This section discusses several possible choices of the bounding parameter $\pi$ for the anticipation probability among treated units under different setups.

\textbf{Example 1} $\pi=\pi_0$, where $\pi_0$ is a constant number. If $\pi$ is a constant number, a common upper bound exists for the possibility of anticipation. This choice of $\pi$ may be consistent with the setup where people get treated randomly and receive private information that helps with anticipation. Then, the overall anticipation possibility should be no more than the proportion of people who have access to this private information.

\textbf{Example 2} $\pi=\mathbb{P}[D_i=1]$.
This example states that the possibility of people within the treatment group anticipating does not exceed the proportion of people treated at last. This argument follows the idea that anticipation will happen when a future treatment sends some signals and unobservable information in advance. These signals are the bases for someone anticipating a treatment. Suppose the density of these signals caused by future adoptions of policies is related to the overall scope of the treatment. In that case, using the treated probability to help bound the proportion of people who anticipate is reasonable. In the appendix, I explain this choice and corresponding assumptions in a model where people anticipate from public information.

\textbf{Example 3} The univariate bound can be modified to incorporate the idea of stratification. Suppose researchers are willing to divide units into several subgroups and allow anticipation behavior to differ among subgroups. In that case, I can choose $\pi$ as a $k$-dimensional vector if I have $k$ subgroups in total. For instance, the vector of assignment can be summarized according to genders or geographical areas, and researchers can get a bound separately for each subgroup. This choice of $\pi$ can also be regarded as a bound conditional on a discrete variable that divides the group based on several categories, and can link to the case with covariates.

\textbf{Example 4} Suppose anticipation behavior happens among known reference groups for each unit, as mentioned in \cite{manski2013identification}. In that case, I can choose $\pi$ based on subgroup information. For example, if researchers would like to use the treatment ratio to capture the density of information, and on the other hand, they also believe this kind of interaction only happens among units within a specific geographical distance, they can choose $\pi$ as the treatment ratio for each subgroup defined by the given geographical distance.

Additionally the choice of $\pi$ can also play the role of sensitivity analysis. The expression of the bounds should be monotonic in $\pi$, and researchers can use different choices of $\pi$ to explore the robustness of obtained conclusions by checking the specific cutoff under which the consistency of conclusion can be maintained. This sensitivity analysis also helps us understand to what extent the conclusion depends on the choice of bounds, and researchers can report the range of anticipation probability that rejects a particular null hypothesis.
\section{Estimation and Inference}
The previous section illustrates that by using a DID approach, the treatment effect for treated with anticipation is partially identified under certain assumptions. The population average expressions of the interval bounds lead to straightforward estimators using sample means under independent assumptions. This section builds uniformly effective confidence sets for the partially identified parameters.

Assume researchers observe data from a distribution $\text{P}\in\mathbf{P}$ with the unobservable parameter, $\mathbb{P}[A_i=1|D_i=1]\in[0,\pi]$. $\mathbf{P}$ refers to the family of distributions that satisfy the sampling, potential outcomes restrictions. For the inferential goal under partial identification, I build a confidence set that is uniformly consistent in level $\alpha$, namely
\[\lim_{n\to\infty}\inf_{\text{P}\in\mathbf{P},\mathbb{P}[A_i=1|D_i=1]\in[0,\pi]}\mathbb{P}[\mu_g\in CS_{\alpha}^{\mu}]\ge\alpha,\]
where $CS_{\alpha}^{\mu}$ is the $\alpha$ level confidence set for the parameter of interest $\mu_g$.

Here, I provide confidence sets based on \cite{imbens2004confidence}, \cite{stoye2009more}, and \cite{stoye2020simple}. The upper and lower bounds for the identified set are estimated using the same sample and are thus highly correlated. I can build an easy-to-implement confidence set by modifying the method addressed above. For notation simplicity, refer to the upper and lower bounds of parameter $\mu_g$ as $\mu_{g,u}$ and $\mu_{g,l}$. A uniformly effective confidence set can be built if the corresponding estimators $\hat{\mu}_{g,u}$ and $\hat{\mu}_{g,l}$ exist and satisfy the following assumptions
\begin{assumption}
	\label{asymnorm}
	$\sqrt{n}\left(\begin{array}{c}
		\hat{\mu}_{g,l}-\mu_{g,l} \\ 
		\hat{\mu}_{g,u}-\mu_{g,u}
	\end{array}\right)\xrightarrow{d}\mathcal{N}\left( \left[ 
	\begin{array}{c}
		0\\ 
		0 %
	\end{array}%
	\right] ,\left[ 
	\begin{array}{cc}
		\sigma _{l}^{2} & \rho \sigma _{l}\sigma _{u} \\ 
		\rho \sigma _{l}\sigma _{u} & \sigma _{u}^{2}%
	\end{array}%
	\right] \right)$ uniformly in $\text{P}\in\mathbf{P}$, and estimators $(\hat{\sigma}_l^{2},\hat{\sigma}_u^{2},\hat{\rho})$ converge to their population values uniformly in $\text{P}\in\mathbf{P}.$
\end{assumption}
\begin{assumption}
	\label{finitevariance}
	For all $\text{P}\in\mathbf{P}$, $\underline{\sigma}^{2}\le\sigma_l^{2},\sigma_{u}^{2}\le\bar{\sigma}^{2}$ for some positive and finite $\underline{\sigma}^{2}$ and $\bar{\sigma}^{2}$ and $\mu_{g,u}-\mu_{g,l}=\Delta\le\bar{\Delta}<\infty$.
\end{assumption}
\begin{theorem}	\label{infer}
	Under Assumption \ref{asymnorm} and \ref{finitevariance}, define $\hat{\sigma}=\max\{\hat{\sigma}_l,\hat{\sigma}_u\}$ and  find $C_n$ that satisfies
	\[\Phi\left(C_n+\sqrt{n}\frac{\hat{\mu}_{g,u}-\hat{\mu}_{g,l}}{\hat{\sigma}}\right)-\Phi(-C_n)=\alpha.\]
	$\Phi$ represents the cumulative distribution function for standard normal distribution.\\
	Then, I have	
	\[\lim_{n\to\infty}\inf_{\text{P}\in\mathbf{P},\mathbb{P}[A_i=1|D_i=1]\in[0,\pi]}\mathbb{P}\left(\mu_g\in \left[\hat{\mu}_{g,l}-C_n\frac{\hat{\sigma}}{\sqrt{n}},\hat{\mu}_{g,u}+C_n\frac{\hat{\sigma}}{\sqrt{n}}\right]\right)\ge\alpha.\]
\end{theorem}

To be consistent with the setup in the empirical application, I analyze the case $\tau_g\le0\le\mu_g$ as an example, and I have $\mu_g\in m_g\left[\frac{1}{1+\pi},1\right]$. 
Corresponding bound estimators will be
\begin{eqnarray*}
	\hat{\mu}_{g,u}&=&\frac{1}{n_1}\sum_{i=1}^{n}[g(Y_{i1})-g(Y_{i0})]D_i-\frac{1}{n_0}\sum_{i=1}^{n}[g(Y_{i1})-g(Y_{i0})](1-D_i)\qquad\\ 
	\hat{\mu}_{g,l}&=&\frac{\hat{\mu}_{g,u}}{1+\hat{\pi}}\qquad n_1=\sum_{i=1}^{n}D_i\qquad n_0=n-n_1\qquad\hat{\pi}\ \text{is a consistent estimator for $\pi$.}
\end{eqnarray*}
If one would like to choose $\pi=\mathbb{P}[D_i=1]$, a straightforward $\hat{\pi}$ will be $\frac{1}{n}\sum_{i=1}^{n}D_i$. The standard errors can be found in the supplemental appendix.

The assumptions and results mainly follow \cite{imbens2004confidence} and \cite{stoye2009more}. When compared with the \cite{imbens2004confidence} approach, the confidence set I construct is slightly different in that I choose to extend along with the upper and lower bounds by the same length $C_n\frac{\hat{\sigma}}{\sqrt{n}}$ where \cite{imbens2004confidence} choose the same critical value but the standard errors are different. An intuitive explanation is that although the estimators of the upper and lower bounds are ordered by construction, which is an important assumption mentioned in \cite{stoye2009more}, the upper and lower bounds can have the reverse order and the interval changes from $[\frac{m_g}{1+\pi},m_g]$ to $[m_g,\frac{m_g}{1+\pi}]$ when the confidence set contains both positive and negative values. Therefore, the corresponding variances for the estimators of upper and lower bounds need to be accommodated to use the larger one for both bounds. This modification works for the construction of confidence sets with perfectly correlated and proportional upper and lower bounds, especially when the confidence set contains 0. The proof is discussed in the appendix.

The change in the expression of confidence sets changes the significance level of rejecting the specific null hypothesis, $H_0: \mu_g=0$, in many cases compared with the situation without anticipation. One interesting case worth mentioning happens when $\mu_g$ and $\tau_g$ have different signs. For the specific null hypothesis, $H_0: \mu_g=0$, I can calculate the values of t-statistics that guarantee the conclusion of whether rejecting it or not unchanged regardless of anticipation.
\begin{corollary}\label{null0}
	Suppose $t^{*}$ satisfies \[\Phi(t^{*})-\Phi(-t^{*}/2)=\alpha.\]
	$\Phi$ represents the cumulative distribution function for standard normal distribution, and the inferential goal is to test $H_0:\mu_g=0$ at level $\alpha$. Suppose confidence set $CS_{\alpha}^{\mu}$ is constructed following the procedure in Theorem \ref{infer}. If $\mu_g$ and $\tau_g$ have different signs, and the t-statistic from the DID model without anticipation $\tilde{t}$ satisfies $\lvert\tilde{t}\rvert>t^{*}$; then for any $\pi$, $0\not\in CS_{\alpha}^{\mu}$.
\end{corollary}

Corollary \ref{null0} gives empirical researchers a specific cutoff $t^{*}$ for the most common case of testing $H_0: \mu_g=0$. If the absolute value of the t-statistic exceeds $t^{*}$ for the case of different signs, taking anticipation into consideration will not change the conclusion of rejecting the null hypothesis. For example, when $\alpha=0.95$, the corresponding $t^{*}$ is 3.3. Thus, if the absolute value of the t-statistic is larger than 3.3 without anticipation, the zero hypothesis for the treatment effect can still be rejected regardless of the anticipation probability when the treatment and anticipatory effects have different signs. This corollary gives empirical researchers a cutoff where they can claim the effectiveness of their conclusions even with anticipation as long as the t-statistic is large enough.
\section{Empirical Application}
In this section, I illustrate the results of this paper in the environment established by \cite{fitzpatrick2014early}, which analyzes the effects of an early retirement incentive program on students' achievement. The authors conducted a DID analysis using exogenous variations from the early retirement incentive (ERI) program targeting teachers in Illinois during the mid-1990s to evaluate the effect of large-scale teacher retirements on student achievement. The Teacher Retirement System in Illinois requires retired members who are at least 55 years old and have 20 years of service experience to collect pension benefits at a 6\% discount rate below age 60. If both the employer and employee pay a one-time fee, an Early Retirement Option allows eligible members to collect their full benefits. In 1992-1993 and 1993-1994, an early retirement incentive (ERI) program was offered as an alternative to ERO, which allowed employees to buy five extra years of age and experience as long as they retired immediately. This alternative allowed those at least 50 years old and with 15 years of service credit to increase their retirement benefits. 

Notably, the ERI programs may affect students' learning, because such programs might lead to a change in teachers' experience and age structure, which will eventually influence students' grades. In their paper, the authors used a DID approach to analyze how promoting the ERI program affected students' grades. They found no evidence of an adverse effect and even found a positive effect on grades in some circumstances. I analyze the average treatment effect for treated by taking anticipation into consideration. The outcome of interest is students' grades, and the major difference is teachers might now anticipate the program in advance and benefit from it.

The authors collected data from several sources. The Teacher Service Record is an administrative dataset that contains information on employees from Illinois Public Schools. The second set of data provides school-level information on test scores for given subjects and grades. The third source contains demographic information of students in schools. The analysis is restricted to teachers of third, sixth and eighth grades, because standardized testing in Illinois focuses on these grades. One major issue for the data is the ERI take-up is not observed directly. The authors exploited the fact that teachers with 15 or more years of experience were most likely to take up the program, and used it as a proxy for the intensity of treatment by the ERI program.

Consider the restrictions I impose on potential outcomes for the case with anticipation. I require that the students' grades in classes whose teachers are ineligible or choose not to retire early should not be affected, and I also require if teachers are not aware of this program in advance, we should see no change in the grades. Further, I require that once the ERI program is implemented, whether teachers anticipate or not should no longer affect the students' grades. The parallel trends assumption requires that trends in students' grades among schools with fewer treated teachers are precise counterfactuals for trends among schools with more treated teachers without anticipation.

Following the idea that teachers, regardless of eligibility, may have some information from a third party before the implementation of the ERI program so that they may have anticipated something, I choose $\pi=\mathbb{P}[D_i=1]$, and the probability of getting treated is estimated by calculating the proportion of experienced teachers with more than 15 years of service credit, and I get correpsonding $\hat{\pi}$. Recall this bound is used to capture the intensity of potential unobservable information, which is proportional to the intensity of treatment, and I can use the proportion of teachers with more than 15 years of teaching experience to bound the anticipation probability. The magnitude effect assumption requires that the anticipatory effect, $\tau$, which is the result of potential behavior changes of teachers who think they can retire early, has a smaller magnitude than the treatment effect, $\mu$, which is the change in students' grades caused by the ERI program after its implementation. Even under a perfect anticipation setup from econometricians' perspective, the teachers are not confident they would be eligible for the policy and take it up; thus the expectation that the anticipatory effect does not have a larger magnitude than the treatment effect when the policy occurs is reasonable. Further, the expectation that the treatment effect $\mu$ has the same sign as the non-negative DID estimator from \cite{fitzpatrick2014early} is reasonable. On the other hand, I follow the argument in the same paper that claims teachers near the retirement age and who anticipate the possibility of early retirement may exert less effort than younger teachers. Therefore, one can reasonably argue the anticipatory effect $\tau$ is non-positive.

With all the assumptions discussed above, I can analyze the treatment effect with anticipation starting from the following equation in \cite{fitzpatrick2014early} that estimates the DID estimator
\[Y_{igt}^{s}=\beta_0+\beta_1(Teachers\ge 15)_{ig}\times Post_t+\beta_2Teachers_{ig}\times Post_t+\gamma\bm{X}_{it}+\delta_{ig}+\varphi_{tg}+\varepsilon_{itg}^{s}.\]
$Y_{igt}^{s}$ is the test score of grade $g$ for subject $s$ in school $i$ and year $t$. $Teachers\ge15$ is the number of teachers with at least 15 years of experience before 1994 and who are thus eligible for the program. $Teachers$ is the average total numbers of teachers. $Post$ serves as the period, an indicator variable that equals 1 after the school year of 1993. Vector $\bm{X}$ contains demographic information, and $\delta$ and $\varphi$ are corresponding fixed effect terms. Although covariates are included here, the parametric assumption that it enters the outcome linearly implies the treatment effect is homogeneous across different values of controls. The intensity of information related to the choice of $\pi$ has already been captured by the proportion of experienced teachers. In the absence of anticipation, $\beta_1$ from this equation estimates the effect of the ERI program on students' grades. Based on the estimator for $\beta_1$ and $\pi$ that I choose, I can analyze the results with anticipation. I check the results for different grades and subjects and compare the cases for all teachers. Results are shown in Table \ref{all}. Similar results using data from subject-specific teachers are listed in the supplemental appendix.

\begin{table}[htbp]
	\centering
	\begin{threeparttable}
		\caption{Effects of the Early Retirement Incentive Program on Test Scores}
		\begin{tabular}{ccccc}
			\toprule
			\toprule
			& \multicolumn{2}{c}{Original Results} & \multicolumn{2}{c}{With Anticipation} \\
			\cmidrule{2-5}          & Math  & Reading & Math  & Reading \\
			\midrule
			All Grades & 0.003 & 0.009 & [0.002,0.003] & [0.006,0.009] \\
			& (0.004) & (0.003) &       &  \\
			& [-0.003,0.010] & [0.002,0.015] & [-0.004,0.010] & [0.001,0.014] \\
			Grade 3 & 0.002 & -0.009 & [0.001,0.002] & [-0.018,-0.009] \\
			& (0.01)  & (0.008) &       &  \\
			& [-0.017,0.021] & [-0.025,0.008] & [-0.018,0.021] & [-0.05,0.022] \\
			Grade 6 & -0.0001 & 0.006 & [-0.0002,-0.0001] & [0.004,0.006] \\
			& (0.005) & (0.004) &       &  \\
			& [-0.01,0.01] & [-0.003,0.015] & [-0.022,0.021] & [-0.004,0.014] \\
			Grade 8 & 0.005 & 0.013 & [0.003,0.005] & [0.008,0.013] \\
			& (0.005) & (0.005) &       &  \\
			& [-0.005,0.015] & [0.004,0.022] & [-0.006,0.014] & [0.001,0.021] \\
			\bottomrule
		\end{tabular}%
		
		\begin{tablenotes}
			\small 
			\item\textbf{Notes:} This table contains data for all teachers. Each column presents results from a separate regression. Teachers who teach multiple grades are included in each grade. Teachers who teach in self-contained classrooms are assumed to teach both math and English. I list identified sets in the first row and 95\% level confidence sets in the third row for each result with anticipation. For comparison purposes, I also provide estimators, standard errors, and 95\% confidence intervals for results from \cite{fitzpatrick2014early}. Standard errors are displayed with parentheses.
		\end{tablenotes}
		\label{all}%
	\end{threeparttable}
\end{table}%

For the partial identification results, I provide identified sets as well as the 95\% confidence sets. From the initial results in \cite{fitzpatrick2014early}, I notice the estimator is at time negative. Because these negative estimates are all insignificant at the $95\%$ level, I conclude this distortion error is due to the finite sample bias. For these estimators, I obtain the identified sets and confidence sets by changing the sign restriction. I find the confidence sets, after I incorporate anticipation, still cannot reject the null hypothesis $\mu=0$, regardless of the sign I choose. The result changes mainly in two aspects. On the one hand, the results with anticipation suggest the treatment effect can be smaller than the one we get directly from the DID approach, because the DID estimator also captures the pre-treatment negative effect caused by anticipation. The effect can be overestimated up to about 30\% because of anticipation. On the other hand, the confidence sets, compared with the DID approach, are slightly shifted leftwards, and this result also reminds people to be more careful when interpreting the non-negative treatment effect. Despite these differences, results incorporating anticipation still support the conclusion that the ERI programs have a non-negative effect on students' grades. These results imply incorporating anticipation can make the result more robust and still support our idea of the non-negative effect of ERI programs on student achievement.

I conduct a robustness check to see the range of choices for $\pi$ that keeps the significance of the estimator at a 95\% level, and show the result in Figure \ref{Robustness Analysis}. I focus on the effect of the early retirement incentive program on the reading grade in grade 8. I present the identified set as well as the 95\% confidence set for a sequence of $\pi$, including 0.1, 0.25, 0.5, $\mathbb{P}[D_i=1]$, 0.75, and 0.9. The shorter interval represents the identified set, and the longer one represents the confidence set. I observe that at an anticipation probability of 0.75, more precisely around 0.7, the confidence set marginally contains point 0, which means this positive treatment effect is quite robust even when taking anticipation into consideration. The null hypothesis will only be rejected when about three fourths of the target teachers anticipate it.
\begin{figure}
	\begin{tikzpicture}[scale=1.3],
		title = Figure II. 3,
		
		\begin{axis}[
			height=8cm,
			width=10cm,  
			ymax=2.2,
			ymin=-0.2,
			xmin=-1,
			xmax=13,
			axis y line*=left,
			axis x line*=bottom,
			xlabel=$\pi$: maximum proportion of anticipators in the treated group,
			ylabel=Grades ($\times10^{-2}$),
			xticklabels={0,0.1,0.25,0.5,$\hat{\pi}$,0.75,0.9},xtick={0,1.0,2.5,5.0,5.7,7.5,9.0},
			x tick label style={anchor=north}]
			\addplot+[only marks][error bars/.cd,y dir=both, y explicit]
			coordinates {
				(0,1.3) +- (0.9,0.9)			
			};
			
			\addplot[red, solid,mark=-]
			coordinates {
				(1.0,1.3) (1.0,1.181)
			};
			\addplot[red, solid,mark=-]
			coordinates {
				(2.5,1.3) (2.5,1.04)
			};
			
			\addplot[red, solid,mark=-]
			coordinates {
				(5.0,1.3) (5.0,0.87)
			};
			\addplot[red, solid,mark=-]
			coordinates {
				(5.7,1.30) (5.7,0.80)
			};
			\addplot[red, solid,mark=-]
			coordinates {
				(7.5,1.3) (7.5,0.746)
			};
			\addplot[red,solid,mark=-]
			coordinates {
				(9.0,1.3) (9.0,0.69)
			};
			\addplot[black,solid,mark=x]
			coordinates {
				(1.3,2.14) (1.3,0.34)
			};
			\addplot[black, solid,mark=x]
			coordinates {
				(2.8,2.10) (2.8,0.25)
			};
			\addplot[black, solid,mark=x]
			coordinates {
				(5.3,2.074) (5.3,0.104)
			};
			\addplot[black, solid,mark=x]
			coordinates {
				(6.0,2.071) (6.0,0.071)
			};
			\addplot[black, solid,mark=x]
			coordinates {
				(7.8,2.065) (7.8,-0.005)
			};
			\addplot[black, solid,mark=x,mark options={solid},]
			coordinates {
				(9.3,2.06) (9.3,-0.07)
			};
			\addplot[dashed] coordinates {(-1.0,0) (10.0,0)};
		\end{axis}
		\begin{customlegend}[legend entries={95\% CI for Point Estimator, Identified Set, 95\% Confidence Set}, legend style={at={(9,6)},anchor=north}]
			\addlegendimage{blue,fill=black!50!blue,mark=*,sharp plot}
			\addlegendimage{red,fill=red,sharp plot,mark=-}
			\addlegendimage{black,fill=black,sharp plot,mark=x}
		\end{customlegend}
	\end{tikzpicture}%
	\caption{Robustness Analysis}
	\label{Robustness Analysis}
\end{figure}
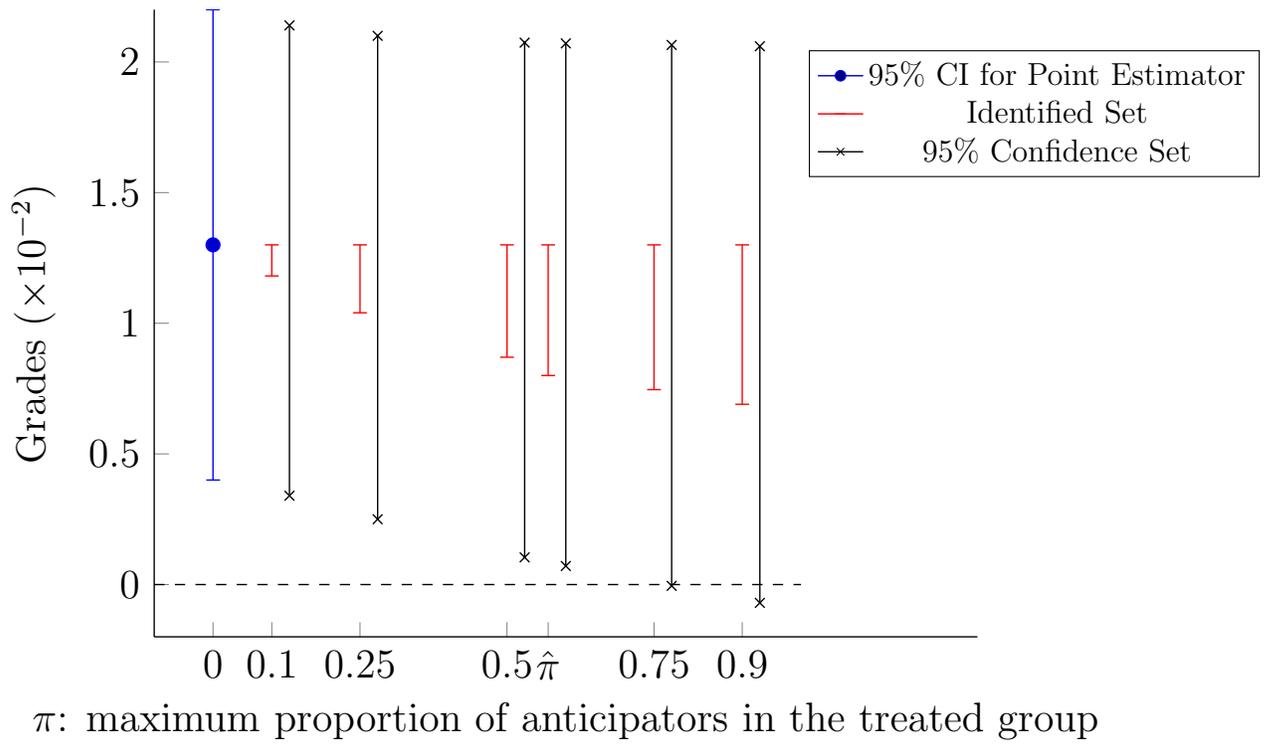
\section{Discussion}
This section discusses modifications of the two-period DID model that only considers the pre-treatment anticipatory behavior in the treated group to incorporate anticipation in broader setups. These generalizations build the anticipation framework on more empirical related assumptions and cover problems researchers encounter in applied work.

First, we focus on the restrctions on pre-treatment anticipatory behavior in the treated group only. In our discussion about Assumption \ref{onlychannel}, I noted one can understand the focus only on anticipatory behavior within the treated group as an implicit assumption of ``perfect anticipation'' which implies units that anticipate will get the anticipated treatment status in the future. However, this assumption might be too strong in some circumstances. For example, people may anticipate the existence of a specific policy, but they are not sure they will get treated. In the early retirement incentive program example, teachers can anticipate the possibility of early retirement, but they are not sure about the amount of service credit they can buy and thus cannot perfectly anticipate their future treatment status. I consider the consequences if a mistake is made when anticipating a future treatment in this section and successfully incorporate the anticipatory behavior in the control group. Furthermore, exploring the robustness of our conclusion by checking the error rate under which consistent conclusions can still be obtained is essential. If researchers aim to test a particular null hypothesis, they can also report the lowest error rate at which the null hypothesis is no longer rejected.

To distinguish between anticipated treatment status and the actual treatment units received and to incorporate imperfect anticipation, I now define the anticipation status variable $A_i$ as a random variable that takes three values $\{-1,0,1\}$. The difference between $A_{i}=0$ and $A_{i}\ne0$ distinguishes between those who anticipate and those who don't. However, among those who anticipate, $A_{i}=1$ indicates this anticipates correctly, whereas $A_{i}=-1$ indicates incorrect anticipation. The potential outcome for unit $i$ in period $t$ still depends on both anticipation and treatment $(a,d)$ and is denoted by the random variable $Y_{it}(a,d)$. The only difference is now $a\in\{-1,0,1\}$ and $A_i$ is no longer a binary treatment. Compared with the benchmark model, some modifications need to be made on the assumptions to incorporate imperfect anticipation.

The random sampling assumption remains unchanged, and the only modification is that anticipatory behavior also happens within the control group, so we have more potential outcomes now and we need another index in the expression of pre-treatment potential outcomes for the control group as well.
\begin{assumption}
	\label{sampling:IA}
	$\{Y_{i0}(a,d), Y_{i1}(d), D_{i},A_{i}\}_{i=1}^{n}$  are independently and identically distributed across $i$.
\end{assumption}

The assumption that requires anticipation to be the only channel for the future to affect the present is still needed. However, the fact that now one's pre-treatment behavior is affected by his anticipated status, which is likely to be different from the treatment status he receives in the future, needs to be addressed.
\begin{assumption} 
	\label{onlychannel:IA}
	The potential outcomes satisfy
	\[Y_{i0}(0,0)=Y_{i0}(0,1)=Y_{i0}(1,0)=Y_{i0}(-1,1)\qquad Y_{i0}(-1,0)=Y_{i0}(1,1).\] 
\end{assumption}

Assumption \ref{onlychannel:IA} mainly describes two groups of units. The first group either does not anticipate or they anticipate they will not get treated in the future so they will behave the same way. People in the second group expect that they will get treated in the future and they will behave the other way. Because anticipation is the only way I allow the future to affect the present, one's anticipated treatment status determines their pre-treatment behavior. A treated person with incorrect anticipation should have the same anticipated treatment status as an untreated person who anticipates correctly, and thus, they should behave the same way as those who do not anticipate. However, those who will get treated and anticipate correctly should behave the same way as those who won't be treated but anticipate incorrectly, because they all think they will be covered. This assumption points out that what drives people's pre-treatment behavior is their beliefs about the treatment status. Trying to distinguish between anticipated treatment status and real treatment received is important in the case in which people make mistakes while anticipating.

One may argue that in the situation of imperfect anticipation, people may no longer clearly anticipate the future treatment and may believe they will get treated with a probability. Different people hold different beliefs about their treatment possibilities and behave differently. This situation can be regarded as a case in which anticipation is a multivalue treatment, and people with different beliefs receive different levels of anticipation treatment. Because all things related to anticipation are unobservable, introducing more levels of different anticipation treatments also requires more assumptions regarding each group. Therefore, I still focus on the case in which people's anticipation about the future concerns whether they will be treated.

The parameter of interest $\mu_g$ is still
\[\mu_g=\mathbb{E}[g(Y_{i1}(1))-g(Y_{i1}(0))|D_i=1],\]
with similar restrictions on $g(.)$ function. Because people's anticipation status does not affect their post-treatment outcomes, I still use one index to represent the potential outcomes $Y_{it}(d)$. The anticipatory effect is modified as
\begin{eqnarray*}
	\tau_g &=&\mathbb{E}[g(Y_{i0}(1,1))-g(Y_{i0}(0,0))|D_i=1,A_i=1]\\
	&=&\mathbb{E}[g(Y_{i0}(-1,0))-g(Y_{i0}(0,0))|D_i=1,A_i=-1],
\end{eqnarray*}
where I implicitly require that the anticipatory effects for those who anticipate correctly and incorrectly are the same.
\begin{assumption}
	\label{paralleltrends:IA} \[\mathbb{E}[g(Y_{i1}(0))-g(Y_{i0}(0,0))|D_i=1]=\mathbb{E}[g(Y_{i1}(0))-g(Y_{i0}(0,0))|D_i=0].\]
\end{assumption}

The key idea of the parallel trend is to require those who get treated and those who do not behave in the same way without the treatment. Following this idea, I need to choose those who have an anticipated untreated status when compared with the outcome in the first period.
\begin{assumption}
	\label{prob_bound:IA} 
	\[\mathbb{P}[A_i\ne0|D_i=1],\mathbb{P}[A_i\ne0|D_i=0]\le \pi\]
	\[\mathbb{P}[A_i=-1|D_i=1,A_i\ne0]=\mathbb{P}[A_i=-1|D_i=0,A_i\ne0]=\varepsilon.\]
\end{assumption}

The first part of Assumption \ref{prob_bound:IA} now requires us to choose a $\pi$ as the bound for the probability of anticipation among treated and control groups. This modification on restriction is straightforward because under the ``perfect anticipation'' situation, units that won't get treated will not react to the anticipation, but now they may react to it because of the incorrect anticipation. If one would like to argue a specific relationship exists between the possibility of anticipation within treated and control groups, this assumption might be relaxed. For now, I assume a common bound $\pi$ for two probabilities is chosen. In the second part, I assume the fraction of units that anticipate incorrectly is known as $\varepsilon$ across treated and control groups. Recall that in the discussion about the bias caused by anticipation, I point out that the bias is driven by those who anticipate and react to it in the first period. Under the setup of imperfect anticipation, the proportion of units that cause the bias is determined by anticipated treatment status and thus related to both the proportion of those who anticipate and the accuracy rate among anticipators.
\begin{assumption}
	\label{magnitude:IA}
	$\abs{\tau_g}\le\abs{\mu_g}$.
\end{assumption}

Assumption \ref{magnitude:IA} is the same magnitude restriction as before. Based on the assumptions above, we are now able to partially identify the parameter of interest $\mu_g$.
\begin{theorem}\label{imperfect}
	Under Assumptions \ref{sampling:IA}-\ref{magnitude:IA}, the parameter of interest $\mu_g$ is partially identified via a closed interval based on  $m_g=\mathbb{E}[g(Y_{i1})-g(Y_{i0})|D_i=1]-\mathbb{E}[g(Y_{i1})-g(Y_{i0})|D_i=0]$, $\pi$ and $\varepsilon$. Define $\mu_{g,1}(\varepsilon)=\frac{m_g}{1+\text{sgn}(\tau_g\mu_g)\pi\varepsilon}$ and $\mu_{g,2}(\varepsilon)=\frac{m_g}{1-\text{sgn}(\tau_g\mu_g)\pi(1-\varepsilon)}$. Then we have, \[\mu_g\in\left[\min\left\{\mu_{g,1}(\varepsilon),\mu_{g,2}(\varepsilon)\right\},\max\left\{\mu_{g,1}(\varepsilon),\mu_{g,2}(\varepsilon)\right\}\right].\]
\end{theorem}

If $\varepsilon=0$, which represents the situation of perfect anticipation, this interval degenerates to the interval we get in the benchmark model for the treatment effect. Compared with the interval I get for the treatment effect above, one significant difference is that now the treatment effect is not bounded by the DID estimator from one side. That difference derives from the fact that both the control group and the treated group deviate in the first period. For a perfect anticipation setup, those in the control group will not react to the anticipation, and only the treated group is moving either upward or downward depending on the signs of anticipatory effect. When imperfect anticipation is allowed and both control and treated groups react to it, the distortion in the first period can be either positive or negative depending on the anticipation possibility in different groups. To accommodate this framework in more empirical settings, I don't impose specific assumptions on the relationship between the possibility of anticipating in different groups. If in specific situations, for example, a case in which those who get treated receive private information, researchers can decide the relationship between these two possibilities, improving the identified set based on further assumptions is possible.

Although I start with the case in which $\varepsilon$ is known, a better explanation for including the error rate of anticipation is to understand this procedure as a sensitivity check. Researchers can choose different possible error rates $\varepsilon$ and analyze the region where their conclusions are robust to the choice of error rate. Further, if a specific null hypothesis is tested, the error rate among which the conclusion holds consistently can also be reported. This robustness check procedure helps people understand to what extent the conclusion is affected by the assumption of perfect anticipation.

One might also be interested in what happens if the anticipatory behavior in the treated group has an effect on the post-treatment behavior and the post-treatment potential outcomes in the treated group are also different between those who anticipate and those who do not. Starting from the benchmark model, now let us assume $Y_{i1}(0,1)$ and $Y_{i1}(1,1)$ are different. Then, we have two effects related to anticipation: $\tau_1=\mathbb{E}[g(Y_{i0}(1,1))-g(Y_{i0}(0,1))|D_i=1]$ and $\tau_2=\mathbb{E}[g(Y_{i1}(1,1))-g(Y_{i1}(0,1))|D_i=1]$. Following logic similar to that above, one can find $\mu_g=m_g+\mathbb{P}[A_i=1|D_i=1](\tau_1-\tau_2)$, which implies that if no more assumptions about the relationship between the pre- and post-treatment effect caused by anticipation are imposed, nothing more can be said. Further, if one would like to assume the effect caused by anticipation remains unchanged before and after the treatment, this equation points out that the existence of anticipation will have no effect on the identification and estimation of the parameter of interest here. This assertion is a generalization of the no anticipation assumption in the canonical DID model where both effects are assumed to be zero.

Further generalizations including the model incorporates anticipation with covariates in multiple periods as well as nonlinear outcomes that involves the change-in-changes model are also provided in the appendix.
\section{Conclusion}
This paper proposes a potential-outcomes framework for analyzing treatment effects in the presence of anticipatory behavior. Based on a two-period DID model, the findings of this paper show how the standard estimator can be biased, and provide a weighted average of the treatment effect and the anticipatory effect. I also provide conditions under which I can obtain upper and lower bounds for the treatment effects. The motivation and implications of each assumption are discussed to accommodate empirical research backgrounds. This paper contributes to the empirical research by introducing anticipation in a practical and easy-to-generalize way starting from the classical DID model, which makes it robust to the existence of this kind of forward looking behavior. An easy-to-implement estimation and inference strategy is also provided. Based on this inference procedure, I propose a sensitivity analysis approach that discusses the validity of conclusions under different restrictions on the anticipation possibility, and this approach suggests a specific range of t-statistics that guarantee the effectiveness of the conclusion without anticipation when the treatment and anticipatory effect have different signs. I illustrate the results in this paper by examining the effect of early retirement incentive programs on student achievement while considering anticipation, and show potential pitfalls if anticipation is ignored. To make this framework more general and less restrictive, I provide several modifications based on the two-period DID model to be consistent with common empirical setups.

The analysis for this paper still leaves open questions, some of which are discussed in the appendix. I provide several alternative combinations of assumptions that can be used to obtain partial identification results for treatment effects with anticipation, for example, bounded outcomes assumptions and further conditional independence restrictions on potential outcomes and treatments. The choice of these assumptions depends on the empirical backgrounds researchers are working on. Alternative assumptions combined with available bounds provide applied workers with more choices that fit into broad applied circumstances. Further work can focus on some frequent issues in empirical studies, for example, anticipation effects related to instrumental variables. When a time gap exists between the instrumental variable and the treatment, the instrumental variable can cause people to anticipate future treatment and thus react to it before the treatment occurs. This setup is also related to cases with imperfect compliance and situations where anticipation will affect people's future selections into treatments. Another possible extension is to incorporate the anticipation phenomenon in the synthetic control framework. This extension makes sense because the treatments in typical synthetic control applications are often big policy changes that would naturally be anticipated. Following \cite{ferman2019synthetic}, the anticipation treatment can be regarded as an unobservable confounder that is correlated with treatment, because only those who get treated in the future will react to anticipation. In that case, the pre-treatment weight that fits well may not construct good counterfactual post-treatment outcomes for the treated unit, and thus causes problem. Analyzing the behavior of the synthetic control estimator and DID estimator with anticipation and comparing the performance of these two approaches will be of interest for applied work. By considering these situations, I am more likely to incorporate anticipation in more diverse empirically relevant situations and introduce it into more applied models.
\newline

\textit{Acknowledgements.}  I am deeply grateful to Matias Cattaneo for continued advice and encouragement. I thank Yuehao Bai, Ming Fang, Max Farrell, Yingjie Feng, Zheng Gong, Florian Gunsillius, Andreas Hagemann, Xuming He, Michael Jansson, Shaowei Ke, Ziteng Lei, Xinwei Ma, Kenichi Nagasawa, Mel Stephens, Gonzalo Vazquez-Bare, Yian Yin and seminar participants at many institutions for their valuable feedbacks. All errors are my own.

\newpage
\bibliographystyle{chicago}
\bibliography{Notes}

\newpage
\appendix
\section{Appendix}
\begin{proof}[Proof of Theorem \ref{chap2:benchmark}]
	
	Based on different signs of $\mu_g$ and $\tau_g$, the expression of the upper and lower bounds of the closed intervals might be different. However, the critical logic to obtain the identified set is the same. Here, I take the case in which $0\le\tau_g\le\mu_g$, for example, and all the other cases can be analyzed based on similar logic.
	
	I know
	\begin{eqnarray*}
		\mu_g&=&\mathbb{E}[g(Y_{i1}(0,1))-g(Y_{i1}(0,0))|D_i=1]\\
		&=&\mathbb{E}[g(Y_{i1}(0,1))-g(Y_{i0}(0,0))-(g(Y_{i1}(0,0)-g(Y_{i0}(0,0)))|D_i=1]\\
		&=&\mathbb{E}[g(Y_{i1}(1))|D_i=1]-\mathbb{E}[g(Y_{i0}(0))|D_i=1]\\
		&-&(\mathbb{E}[g(Y_{i1}(0))|D_i=0]-\mathbb{E}[g(Y_{i0}(0))|D_i=0]).
	\end{eqnarray*}
	The first, third, and fourth terms are observable, whereas for the second one,
	\[\mathbb{E}[g(Y_{i0})|D_i=1]=\mathbb{E}[g(Y_{i0}(0,1))|D_i=1]+\mathbb{E}[A_i(g(Y_{i0}(1,1))-g(Y_{i0}(0,1)))|D_i=1]\]
	\[=\mathbb{E}[g(Y_{i0}(0,1))|D_i=1]+\mathbb{P}[A_i=1|D_i=1]\tau_g\le\mathbb{E}[g(Y_{i0}(0))|D=1]+\pi\tau_g.\]
	I further apply the inequality that $0\le\tau_g\le\mu_g$ and get
	\[\mu_g\le\mathbb{E}[g(Y_{i1}(1))-g(Y_{i0})|D_i=1]-(\mathbb{E}[g(Y_{i1}(0))|D_i=0]-\mathbb{E}[g(Y_{i0}(0))|D_i=0])+\pi\tau_g.\]
	For simplicity, I call $m_g=\mathbb{E}[g(Y_{i1})-g(Y_{i0})|D_i=1]-\mathbb{E}[g(Y_{i1})-g(Y_{i0})|D_i=0]$, and I can conclude
	\[0\le\tau_g\le\frac{m_g}{1-\pi}\qquad\tau_g\le\mu_g\le\frac{m_g}{1-\pi}.\]
	From the equation above, I note $\mu_g=\mathbb{P}[A_i=1|D_i=1]\tau_g+m_g\ge m_g$. I can conclude $\mu\in\left[m_g,\frac{m_g}{1-\pi}\right]$. 
\end{proof}

\begin{proof}[Proof of Theorem \ref{imperfect}]
	
	As in the previous situation, I take the situation $0\le\tau_g\le\mu_g$ as an example:
	\begin{eqnarray*}
		\mu_g&=&\mathbb{E}[g(Y_{i1}(1))-g(Y_{i1}(0))|D_i=1]\\
		&=&\mathbb{E}[g(Y_{i1}(1))-g(Y_{i0}(0,0))-(g(Y_{i1}(0)-g(Y_{i0}(0,0)))|D_i=1]\\
		&=&\mathbb{E}[g(Y_{i1}(1))|D_i=1]-\mathbb{E}[g(Y_{i0}(0,0))|D_i=1]\\
		&-&(\mathbb{E}[g(Y_{i1}(0))|D_i=0]-\mathbb{E}[g(Y_{i0}(0,0))|D_i=0]).
	\end{eqnarray*}
	The first and the third terms are observable as $\mathbb{E}[g(Y_{i1})|D_i=1]$ and $\mathbb{E}[g(Y_{i1})|D_i=0]$. For the second one, those who anticipate their treatment status correctly will make a difference:
	\[\mathbb{E}[g(Y_{i0})|D_i=1]=\mathbb{E}[g(Y_{i0}(0,0))|D_i=1]+\mathbb{P}[A_i=1|D_i=1]\tau_g,\]
	whereas for the fourth term, those who wrongly anticipate their future will react to it:
	\[\mathbb{E}[g(Y_{i0})|D_i=0]=\mathbb{E}[g(Y_{i0}(0,0))|D_i=0]+\mathbb{P}[A_i=-1|D_i=0]\tau_g.\]
	Then I have 
	\[\mu_g=\mathbb{E}[g(Y_{i1}(1))|D_i=1]-\mathbb{E}[g(Y_{i0})|D_i=1]+\mathbb{P}[A_i=1|D_i=1]\tau_g\]
	\[-\mathbb{E}[g(Y_{i1}(0))|D_i=0]+\mathbb{E}[g(Y_{i0})|D_i=0]-\mathbb{P}[A_i=-1|D_i=0]\tau_g.\]
	\[\mu_g=m_g+(\mathbb{P}[A_i\ne0|D_i=1](1-\varepsilon)-\mathbb{P}[A_i\ne0|D_i=0]\varepsilon)\tau_g.\]
	Using the bound from Assumption \ref{prob_bound:IA}, I can get that the coefficient before $\tau_g$ belongs to the interval $[-\pi\varepsilon,\pi(1-\varepsilon)]$ and follow the same approach from the proof of Theorem \ref{chap2:benchmark}. I then have that $\mu_g\in\left[\frac{m_g}{1+\pi\varepsilon},\frac{m_g}{1-(1-\varepsilon)\pi}\right]$
\end{proof}

\begin{proof}[A Toy Economic Behavior Model]
	
This section provides a toy economic behavior model that motivates the choice of $\pi=\mathbb{P}[D_i=1]$ and explains what assumptions one needs to infer this choice of bounding parameter. Consider a setup in which a future policy has led to some public information prior to its implementation, and people can take advantage of this kind of signal to anticipate. For example, rumors and changes in teaching assignments may arise before the early retirement incentive program occurs, and teachers may anticipate based on them. The density of the signals is correlated with the overall treatment ratio, because a higher ratio of experienced teachers who are eligible for the program implies more people are interested in the policy, and they may talk more about it.

For simplicity, consider a case in which $\mathbb{P}[A_i=1|D_i=0]=\mathbb{P}[A_i=1|D_i=1]$, which means the probability of anticipating is the same between the control and treatment groups. Suppose the density of information generated by the future implementation of the treatment, denoted by $u$, is proportional to the treatment ratio $\mathbb{P}[D_i=1]$ in the form
\begin{equation*}	
	u=\alpha\mathbb{P}[D_i=1]\qquad\alpha>0,
\end{equation*}
and that this information is known to both treated and control groups. This formula captures the idea that a higher treatment ratio will lead to a case in which the information needed for anticipation is more explicit for people. 

For each unit $i$, I introduce a random variable $U_i$ that represents the level of information to which one needs to be exposed for anticipation to occur. A higher $U_i$ means this unit needs more signals to realize the possible treatment. By contrast, a lower $U_i$ indicates this unit is keenly observant and can reach a conclusion with less information. I will use $F(.)$ to represent the c.d.f of $U_i$ across the population. For any single unit $i$, the anticipatory behavior follows \[A_i=\mathbb{I}[U_i\le u],\]
which means the density of information needs to be larger than the cutoff $U_i$ for unit $i$ to anticipate. Thus, I have
\begin{equation*}
	\mathbb{P}[A_i=1]=\mathbb{P}[U_i\le u]=F(\alpha\mathbb{P}[D_i=1]).
\end{equation*}
I further assume $f^{\prime}(.)\ge0$, where $f(.)$ is the probability density function of the random variable $U_i$. This assumption states that the fraction of people who can marginally anticipate increases with the level of information needed to form anticipation, consistent with the intuition that people who can anticipate based on little information should be small.

Based on the assumptions above, I can conclude \[\frac{\partial^{2}\mathbb{P}[A_i=1]}{\partial\mathbb{P}[D_i=1]^{2}}=\alpha^{2}f^{\prime}(\alpha\mathbb{P}[D_i=1])\ge0,\]
which means $\mathbb{P}[A_i=1]$ should be a convex function on interval $[0,1]$ with respect to $\mathbb{P}[D_i=1]$. If $\mathbb{P}[D_i=1]=0$, I conclude $\mathbb{P}[A_i=1]=0$, because there is nothing to anticipate. I also know that when $\mathbb{P}[D_i=1]=1$, $\mathbb{P}[A_i=1]\le 1$. These bounds, combined with the convex function argument, show \[\mathbb{P}[A_i=1]\le\mathbb{P}[D_i=1]\times1+0=\mathbb{P}[D_i=1].\]	
The above model illustrates that under a setup in which units can get information that is public for all the groups but unobservable for econometricians and these units anticipate future treatments based on the information, Assumption \ref{prob_bound} is satisfied by choosing $\pi=\mathbb{P}[D_i=1]$ as long as one is willing to assume fewer people can anticipate from less information, and as information accumulates, the number of people who marginally learn it increases.
\end{proof}

\begin{proof}[Including covariates]
	
	By taking all the assumptions conditionally, this approach can be generalized to include covariates. I use notation $b(x)$ to represent the value of parameter $b$ when conditional on ``$X=x$.'' Following a similar argument in Theorem \ref{chap2:benchmark}, I focus on the case in which $0\le\tau_g(x)\le\mu_g(x)$.
	From the proof of Theorem \ref{chap2:benchmark}, I know
	\[\mu_g(x)=\mathbb{E}[g(Y_{i1}(1))|D_i=1,X]-\mathbb{E}[g(Y_{i0}(0))|D_i=1,X]\]
	\[-(\mathbb{E}[g(Y_{i1}(0))|D_i=0,X]-\mathbb{E}[g(Y_{i0}(0))|D_i=0,X]).\]
	The first, third, and the fourth terms are observable, whereas for the second one,
	\begin{eqnarray*}
	\mathbb{E}[g(Y_{i0})|D_i=1,X]&=&\mathbb{E}[g(Y_{i0}(0,1))+A_i(g(Y_{i0}(1,1))-g(Y_{i0}(0,1)))|D_i=1,X]\\
	&=&\mathbb{E}[g(Y_{i0}(0,1))|D_i=1,X]+\mathbb{P}[A_i=1|D_i=1,X]\tau_g(x)\\
	&\le&\mathbb{E}[g(Y_{i0}(0))|D_i=1,X]+\pi(x)\tau_g(x).
	\end{eqnarray*}
	I further apply the inequality that $0\le\tau_g(x)\le\mu_g(x)$, and call $m_g(x)=\mathbb{E}[g(Y_{i1})-g(Y_{i0})|D_i=1,X]-\mathbb{E}[g(Y_{i1})-g(Y_{i0})|D_i=0,X]$. I get
	\[\tau_g(x)\le\mu_g(x)\le m_g(x)+\pi(x)\tau_g(x),\]
	and I can conclude
	\[0\le\tau_g(x)\le\frac{m_g(x)}{1-\pi(x)}\qquad\tau_g(x)\le\mu_g(x)\le\frac{m_g(x)}{1-\pi(x)}.\]
	I note $\mu_g(x)=\mathbb{P}[A_i=1|D_i=1,X]\tau_g(x)+m_g(x)\ge m_g(x)$ and can conclude $\mu_g(x)\in\left[m_g(x),\frac{m_g(x)}{1-\pi(x)}\right]$, or I can write it as 
	\[\mathbb{E}[\mu_g-m_g|X]\ge0\qquad\mathbb{E}\left[\mu_g-\frac{m_g}{1-\pi(x)}\bigg|X\right]\le0.\]
	To avoid calculating so many conditional expectations, I follow \cite{abadie2005semiparametric} and write		\[m_g(x)=\mathbb{E}[\rho_0(g(Y_1)-g(Y_0))|X=x]\qquad\rho_0=\frac{D_i-\mathbb{P}[D_i=1|X]}{\mathbb{P}[D_i=1|X](1-\mathbb{P}[D_i=1|X])}.\].
\end{proof}

\begin{proof}[Multiple Periods] 
	
	In applied work, the standard two-period DID model is not the majority, and people would like to incorporate data from multiple periods to support their conclusions. Considering the case with multiple periods, besides the two-period model, is crucial to accommodate anticipation in longitudinal data. I follow the framework of \cite{sun2020estimating} and consider anticipation in a staggered adoption case with multiple periods and possibly different treatment times. \cite{sun2020estimating} analyze this framework to address the effect of treatment effect heterogeneity in a two way fixed effects model and propose a DID estimator. My focus is on how to incorporate anticipation in a multi-period model and if homogeneity assumptions are imposed and researchers use an alternative approach, a similar logic can be applied to incorporate anticipation.
	
	Suppose the potential outcome for unit $i$ in period $t$ is represented by $Y_{it}(a,e)$. $i\in\left\{1,2,\cdots,n\right\}$ indexes the unit and $e\in supp(E_i)=\left\{1,2,\cdots,T,\infty\right\}$ denotes the date when this unit gets the first treatment and $E_i$ is one realization of $e$. In a setting with staggered adoption, one unit will always get treated after the first treatment. A binary random variable $A_{it}$ that takes value $a\in\{0,1\}$ represents the unobservable anticipation status for unit $i$ in period $t$. $e=\infty$ implies this unit has never been treated and thus belongs to the control group. I still focus on the pre-treatment anticipatory behavior within the treated group, which means I only distinguish $Y_{it}(0,e)$ and $Y_{it}(1,e)$ for $t<e$ and $e\ne\infty$. To incorporate anticipation in the multi-period model, the assumptions in the two-period DID model need some modification.
	\begin{saassumption}
		\label{sample:multi}
		The potential outcomes $\{Y_{it}(a,e), E_i,A_{it}\}_{t=1}^{T}$ are independently and identically distributed across $i$ for $(a,e)\in\{0,1\}\times\{1,2,\dots,T,\infty\}$.
	\end{saassumption}
	\begin{saassumption} $Y_{it}(0,e)=Y_{it}(0,e^{\prime})\quad\text{for}\quad t<\min\{e,e^{\prime}\}$.
		\label{onlychannel:multi}	
	\end{saassumption}
	
	Assumption \ref{sample:multi} restricts the sampling process by imposing i.i.d. restrictions, and Assumption \ref{onlychannel:multi} restricts the anticipatory behavior to be the only reason for the future to affect the present, both in the same way as in the two-period DID model.
	
	Under this setup, the parameter of interest I focus on is 
	\[\mu_g(e,t)=\mathbb{E}[g(Y_{it}(0,e))-g(Y_{it}(0,\infty))|E_i=e]\] with similar restrictions on function $g(.)$. Define the corresponding anticipatory effect for anticipators as \[\tau_g(e,t)=\mathbb{E}[g(Y_{it}(1,e))-g(Y_{it}(0,e))|E_i=e,A_i=1]\] and write $\mu_g(e,t)$ as $\mu(e,t)$ and $\tau_g(e,t)$ as $\tau(e,t)$ if $g(.)$ is the identity function.
	\begin{saassumption}
		For all $t_1\ne t_2$, \[\mathbb{E}[g(Y_{it_1}(0,\infty))-g(Y_{it_2}(0,\infty))|E_i=e]=\mathbb{E}[g(Y_{it_1}(0,\infty))-g(Y_{it_2}(0,\infty))|E_i=\infty]\]
		for all $e$.
		\label{paralleltrends:multi}
	\end{saassumption}

	Parallel trends assumption is imposed and it implies potential outcomes without anticipation and treatment will change in the same way among different groups who get treated at different times.
	
	Following \cite{sun2020estimating}, I build a DID variable to help us analyze the parameter of interest while considering anticipation. Unlike the two-period model, I do not have a natural period 0 for comparison and need to choose one. Consider the following term for $t\ge e$ and $s<e$:
	\[m_g(e,s,t)=\mathbb{E}[g(Y_{it})-g(Y_{is})|E_i=e]-\mathbb{E}[g(Y_{it})-g(Y_{is})|E_i=\infty].\]
	For simplicity, I define $\mathbb{P}[A_{is}=1|E_i=e]=h(e,s)$, which denotes the probability of anticipating in period $s$ before the treatment happens in period $e$. $t$ is a post-treatment period that I am interested in, $e$ is the period when unit $i$ receives the treatment, and $s$ is a chosen pre-treatment period that helps us build the DID estimator as introduced in \cite{sun2020estimating}. I later discuss the choice of period $s$ later. Impose assumptions as in the two-period DID model to get the bounds for the treatment effect.
	\begin{saassumption}
		\label{prob_bound:multi}
		$0\le h(e,s)\le\pi(e,s)$.
	\end{saassumption}

	Based on the discussion in the two-period DID model, I know the source of distortion is those who anticipate and react to it before the treatment. In the multi-period model, thinking about the probability of those who will be treated in period $e$ anticipating in period $s$ is important. Further, I also impose magnitude restrictions on anticipatory and treatment effects in different periods.
	\begin{saassumption}
		\label{magnitude:multi}
		$\abs{\tau_g(e,s)}\le\abs{\mu_g(e,t)}$.
	\end{saassumption}

	In Assumption \ref{prob_bound:multi}, I can choose $\pi(e,s)$ not only based on the treatment period, but also based on the benchmark period $s$ I choose. One common strategy in empirical work to deal with anticipation is to argue anticipation only exists within a certain time length before the treatment, and if people have rich-enough data, they can always drop data when people might anticipate. This argument is equivalent to saying I should choose a far-enough period $s$ and choose $\pi(e,s)=0$ for that period in my setup. However, dropping a subset of data and claiming anticipation disappears before this point is not always reasonable. When I need to consider the bound $\pi(e,s)$, one analogy to the choice of $\mathbb{P}[D_i=1]$ in the two-period model that measures the intensity of treatment is given by the proportion of the group that receives treatment no later than period $e$, $\mathbb{P}[E_i\le e]$, which also captures the idea that people might anticipate from information produced by prior implementations of treatment besides their own future treatments. On the other hand, researchers may want to take into consideration the time gap between the benchmark period $s$ and the treatment period $e$ to capture the idea that the further they are from the treatment, the more difficulty people have anticipating and multiply a time discount factor, for example, $\delta^{e-s}$, for a known discounting factor $\delta\in(0,1)$. Thus, one possible choice of $\pi(e,s)=\delta^{e-s}\mathbb{P}[E_i\le e]$. Researchers can choose different $\pi(e,s)$ based on their situations and empirical backgrounds. 
	
	The choice of $s$ also affects the validity of Assumption \ref{magnitude:multi}, because the relative time gap for the periods $t$ and $s$ to $e$ might differ and possibly affect the strength of this magnitude assumption. Although choosing an $s$ that is considerably far from the treatment period to reduce the anticipation probability and restrict the anticipatory effect magnitude may seem interesting, this choice is not cost-free. Sometimes a long period of pre-treatment data is not available. Even when possible, choosing a benchmark period $s$ that is far from the treatment period requires the parallel trends assumption to hold in a relatively long period and increases the risk of violation. 
	
	Based on the assumptions above, I obtain the following result:
	\begin{satheorem}\label{multi}
		Under Assumptions \ref{sample:multi}-\ref{magnitude:multi} and a chosen s satisfying $s<e$ , the parameter of interest $\mu_g(e,t)$ is partially identified via a closed interval in the following form:  $\mu_g(e,t)\in$
		\[m_g(e,s,t)\left[\min\left\{1,\frac{1}{1-\text{sgn}(\tau_g(e,s)\mu_g(e,t))\pi(e,s)}\right\},\right.\]
		\[\left.\max\left\{1,\frac{1}{1-\text{sgn}(\tau_g(e,s)\mu_g(e,t))\pi(e,s)}\right\}\right]\]
		with $m_g(e,s,t)=\mathbb{E}[g(Y_{it})-g(Y_{is})|E_i=e]-\mathbb{E}[g(Y_{it})-g(Y_{is})|E_i=\infty]$.
	\end{satheorem}
	
	In the multi-period model, I end up with an expression similar to the two-period model with one side being the DID estimator proposed by \cite{sun2020estimating} and the other side enlarging or reducing by a specific ratio depending on the signs of the treatment and anticipatory effect and the bound for anticipation probability. Researchers can still implement a sensitivity check by changing the choice of $\pi(e,s)$ and period $s$ to examine whether the conclusion is robust.
	
	The proof strategy carries from a two-period model, and I continue to take the case in which $0\le\tau(e,s)\le\mu(e,t)$ as an example. For the DID estimand $m_g(e,s,t)$,
	\begin{align*}
		&\mathbb{E}[g(Y_{it})-g(Y_{is})|E_i=e]-\mathbb{E}[g(Y_{it})-g(Y_{is})|E_i=\infty]\\
		&=\mathbb{E}[g(Y_{it})|E_i=e]-\mathbb{E}[g(Y_{is})|E_i=e]-\mathbb{E}[g(Y_{it}(0,\infty))-g(Y_{is}(0,\infty))|E_i=\infty]\\
		&=\mathbb{E}[g(Y_{it}(0,e))+A_{it}(g(Y_{it}(1,e))-g(Y_{it}(0,e)))|E_i=e]\\
		&-\mathbb{E}[g(Y_{is}(0,e))+A_{is}(g(Y_{is}(1,e))-g(Y_{is}(0,e)))|E_i=e]\\
		&-\mathbb{E}[g(Y_{it}(0,\infty))-g(Y_{is}(0,\infty))|E_i=\infty]\\
		&=\mathbb{E}[g(Y_{it}(0,e))|E_{i}=e]+h(e,t)\tau(e,t)-\mathbb{E}[g(Y_{is}(0,e))|E_{i}=e]-h(e,s)\tau(e,s)\\
		&-\mathbb{E}[g(Y_{it}(0,\infty))-g(Y_{is}(0,\infty))|E_i=e]\\
		&=\mathbb{E}[g(Y_{it}(0,e))-g(Y_{it}(0,\infty))|E_{i}=e]-h(e,s)\tau(e,s)\\
		&=\mu_g(e,t)-h(e,s)\tau_g(e,s).
	\end{align*}
	From the equation above, I can introduce Assumptions \ref{prob_bound:multi} and \ref{magnitude:multi} and then get 
	\[m_g(e,s,t)\le\mu_g(e,t)\le\frac{m_g(e,s,t)}{1-\pi(e,s)}.\]
\end{proof}

\begin{proof}[Change-in-Changes Setup]
	
	As an alternative approach to the DID model, \cite{athey2006identification} propose a change-in-changes model that does not depend on the scale of dependent variables and recovers the entire counterfactual distribution of effects of the treatment on the treatment group. The change-in-changes model also incorporates nonlinear potential outcomes. The key identifying assumption is a time-invariant distribution of the unobservable variable across different groups. In this section, I generalize the change-in-changes model to incorporate anticipation.
	
	Following the two-period DID setup, suppose n units $i\in\{1,\dots,n\}$ exist with two periods $t\in\{0,1\}$. Each unit is assigned a binary treatment $D_i$ in the second period and has an unobservable binary anticipation status $A_i$ in the first period. The sampling process and potential outcomes have the same restriction as before. However, to incorporate nonlinear outcomes and not use the parallel trends assumption, I follow the assumptions introduced in \cite{athey2006identification}.
	\begin{saassumption}
		\label{outcome:cic}
		The potential outcome of a unit in the absence of both anticipation and treatment for unit i at period t is affected by an unobservable random variable $U_i$ that represents the unit i's characteristic and satisfies
		\[Y_{it}(0,0)=\varphi(U_i,t)\qquad\text{with}\quad\varphi(u,t)\quad\text{strictly increasing in $u$ for}\quad t\in\{0,1\}.\]
	\end{saassumption} 
	\begin{saassumption}
		\label{para:cic}
		$U\indep t|D$ and $\mathbb{U}_{1}\subseteq\mathbb{U}_0$, where $\mathbb{U}_{t}$ represents the support of random variable $U$ in period $t$.
	\end{saassumption}
	
	Assumptions \ref{outcome:cic} and \ref{para:cic} build the change-in-changes model. Assumption \ref{outcome:cic} requires that the potential outcomes can be captured in a single unobservable random variable $U$, and a higher random-variable value leads to a strictly higher potential outcome. Assumption \ref{para:cic} restricts the population to not change over time, because the random variable that determines the potential outcome is independent of the time period conditional on the group it belongs, which is analogous to the parallel trends assumption in DID model. Based on this assumption, the trend in one group can be used to recover the unobservable potential trend of the other group, and comparison is possible. For this nonlinear setup where one can recover the distribution information of the potential outcomes, the parameter of interest I focus on is the quantile treatment effect for the treated group at a given quantile $q$:
	\[\mu(q)=F_{Y^{1}_{i1}(0,1)}^{-1}(q)-F_{Y^{1}_{i1}(0,0)}^{-1}(q)=F_{Y^{1}_{i1}(1)}^{-1}(q)-F_{Y^{1}_{i1}(0)}^{-1}(q)\]
	$F_{Y}$ represents the cumulative distribution function for the random variable $Y$, and $F^{-1}_{Y}(q)$ represents the q-th quantile of the random variable $Y$. The switch from two indexes to one is based on potential outcome restrictions I impose, because the potential outcomes in the second period should not depend on the anticipation status. For notation simplicity, I use $Y^{d}_{it}$ to represent the distribution of the random variable
	$Y_{it}|D_i=d$, and $Y^{d_2}_{it}(d_1)$ for the random variable $Y_{it}(d_1)|D_i=d_2$, and $Y^{d_2}_{it}(a,d_1)$ for the random variable $Y_{it}(a,d_1)|D_i=d_2$. Under this setup, \cite{athey2006identification} proposes the following identification result.
	\begin{lemma}[\cite{athey2006identification} Theorem 3.1] Suppose Assumptions \ref{sampling}, \ref{onlychannel}, \ref{outcome:cic}, and \ref{para:cic} hold. The distribution of $Y^{1}_{i1}(0)$ can be written as
		\[F_{Y^{1}_{i1}(0)}(y)=F_{Y^{1}_{i0}(0,0)}\left(F^{-1}_{Y^{0}_{i0}}\left(F_{Y^{0}_{i1}}(y)\right)\right).\]
	\end{lemma}
	
	The potential outcome $Y^{1}_{i0}(0,0)$ should have the same distribution as $Y^{1}_{i0}$ without anticipation, and thus, the quantiles of $Y^{1}_{i1}(0)$ and $\mu(q)$ are identified. However, from the previous discussion, we know the existence of anticipation causes  pre-treatment distortions, and the outcomes observed are no longer good measures of potential outcomes in the absence of treatment and anticipation. Following the logic in the previous section, I need some extra assumptions to help bound the quantile treatment effect:
	\[\mu(q)=F_{Y^{1}_{i1}(1)}^{-1}(q)-F_{Y^{1}_{i1}(0)}^{-1}(q)=F_{Y^{1}_{i1}(1)}^{-1}(q)-F^{-1}_{Y^{0}_{i1}}\left(F_{Y^{0}_{i0}}\left(F^{-1}_{Y^{1}_{i0}(0,0)}(q)\right)\right)\]
	The first term can be estimated from the treated group, whereas the second term  $F^{-1}_{Y^{1}_{i0}(0,0)}(q)$ is not identified from the data. Here, I introduce the corresponding quantile anticipatory effect \[\tau(q)=F^{-1}_{Y^{1}_{i0}(1,1)}(q)-F^{-1}_{Y^{1}_{i0}(0,0)}(q)\] and use a similar strategy.
	\begin{saassumption} 
		\label{prob_bound:cic}
		$\mathbb{P}[A_i=1|D_i=1]\le\pi$.
	\end{saassumption}
	\begin{saassumption} 
		\label{magnitude:cic}
		$\abs{\tau(q)}\le\abs{\mu(q)}$.
	\end{saassumption}

	Assumption \ref{prob_bound:cic} is identical to the assumption I impose in the two-period DID model. Assumption \ref{magnitude:cic} is also similar to what has been imposed, and the only difference is that I now restrict the magnitude relationship between the quantile anticipatory effect and quantile treatment effect. Based on the assumptions above, I can build bounds for the quantile treatment effect in the change-in-changes model while incorporating anticipation.
	\begin{satheorem}\label{CIC}
		Under Assumptions \ref{sampling} and \ref{onlychannel} and Assumptions \ref{outcome:cic}-\ref{magnitude:cic}, the parameter of interest, $\mu(q)$, is partially identified via a closed interval. For the sake of notation simplicity, let us define
		\[m(q)=F_{Y^{1}_{i1}(1)}^{-1}(q)-F^{-1}_{Y^{0}_{i1}}\left(F_{Y^{0}_{i0}}\left(F^{-1}_{Y^{1}_{i0}}(q)\right)\right)\]
		\[\phi_u(q) \text{is the closest to zero solution for}\quad F_{Y^{1}_{i1}(1)}^{-1}(q)-F^{-1}_{Y^{0}_{i1}}\left(F_{Y^{0}_{i0}}\left(F^{-1}_{Y^{1}_{i0}}(q)-x\right)\right)-x=0\]
		\[\phi_l(q) \text{is the closest to zero solution for}\quad F_{Y^{1}_{i1}(1)}^{-1}(q)-F^{-1}_{Y^{0}_{i1}}\left(F_{Y^{0}_{i0}}\left(F^{-1}_{Y^{1}_{i0}}(q)+x\right)\right)-x=0.\]
		With the restriction that $\phi_u(q)$ and $\phi_l(q)$ have the same signs as $\mu(q)$, further define
		\[\tilde{\phi}_u(q)=
		\left\{
		\begin{array}{lr}
			F_{Y^{1}_{i1}(1)}^{-1}(q)-F^{-1}_{Y^{0}_{i1}}\left(F_{Y^{0}_{i0}}\left(F^{-1}_{Y^{1}_{i0}}(q-\pi)\right)\right) & q>\pi\\
			+\infty & q\le\pi
		\end{array}
		\right.\]
		\[
		\tilde{\phi}_l(q)=\left\{
		\begin{array}{lr}
			F_{Y^{1}_{i1}(1)}^{-1}(q)-F^{-1}_{Y^{0}_{i1}}\left(F_{Y^{0}_{i0}}\left(F^{-1}_{Y^{1}_{i0}}(q+\pi)\right)\right) & q<1-\pi\\
			-\infty & q\ge1-\pi.
		\end{array}
		\right.\]
		Then we have
		\[0\le\tau(q)\le\mu(q)\qquad\mu(q)\in[m(q),\min\{\phi_u(q),\tilde{\phi}_u(q)\}]\]
		\[\tau(q)\le0\le\mu(q)\qquad\mu(q)\in[\max\{\phi_{l}(q),\tilde{\phi}_{l}(q)\},m(q)]\]
		\[\mu(q)\le0\le\tau(q)\qquad\mu(q)\in[m(q),\min\{\phi_l(q),\tilde{\phi}_u(q)\}]\]
		\[\mu(q)\le\tau(q)\le0\qquad\mu(q)\in[\max\{\phi_{u}(q),\tilde{\phi}_{l}(q)\},m(q)].\]
	\end{satheorem}
	\begin{proof}[Proof of Theorem A. \ref{CIC}]
		
		As before, I take the case $0\le\tau(q)\le\mu(q)$ for an  example. For the lower bound, I know $\tau(q)\ge0$, so I can conclude $F^{-1}_{Y^{1}_{i0}(0,0)}(q)\le F^{-1}_{Y^{1}_{i0}}(q)$ and have 
		\[\mu(q)\ge F_{Y^{1}_{i1}(1)}^{-1}(q)-F^{-1}_{Y^{0}_{i1}}\left(F_{Y^{0}_{i0}}\left(F^{-1}_{Y^{1}_{i0}}(q)\right)\right),\]
		which is identifiable from observable variables. On the other hand, I can build the upper bound for $\mu(q)$ through $F^{-1}_{Y^{1}_{i0}}(q)\le F^{-1}_{Y^{1}_{i0}(0,0)}(q)+\tau(q)\le F^{-1}_{Y^{1}_{i0}(0,0)}(q)+\mu(q)$ using magnitude restriction. Then, the upper bound for $\mu(q)$, represented by $\phi_{u}(q)$, can be solved by the equation
		\[\phi_u(q)= F_{Y^{1}_{i1}(1)}^{-1}(q)-F^{-1}_{Y^{0}_{i1}}(F_{Y^{0}_{i0}}(F^{-1}_{Y^{1}_{i0}}(q)-\phi_u(q))).\]
		When $\phi_u(q)=0$, the left is smaller than or equal to the right hand side. Both sides are increasing functions, so the minimum nonnegative solution should be $\phi_u(q)$. Note $\phi_u(q)$ is not guaranteed to exist in this approach, and it depends on the distribution of all these random variables. For this approach, I do not use the information from Assumption \ref{magnitude:cic}. If I would like to introduce that assumption, I can make an improvement on $\phi_u(q)$ for some circumstances.
		
		If $q\le\pi$, I can do little to improve the bound as the worst case is that all the people who anticipate are in the lowest q percentage and no information about the original distribution without anticipation can be gained. If $q>\pi$, at most $\pi$ of the people exceed the original q-th quantile of $Y^{1}_{i0}(0,0)$ after the anticipation, so I should have that $F^{-1}_{Y^{1}_{i0}}(q-\pi)\le F^{-1}_{Y^{1}_{i0}(0,0)}(q)$, and then I have
		\[\mu(q)\le F_{Y^{1}_{i1}(1)}^{-1}(q)-F^{-1}_{Y^{0}_{i1}}\left(F_{Y^{0}_{i0}}\left(F^{-1}_{Y^{1}_{i0}}(q-\pi)\right)\right)=\tilde{\phi}_u(q).\]
		For the case $0\le\tau(q)\le\mu(q)$, the lower and upper bounds for $\mu(q)$ are
		\[\theta_l(q)=F_{Y^{1}_{i1}(1)}^{-1}(q)-F^{-1}_{Y^{0}_{i1}}\left(F_{Y^{0}_{i0}}\left(F^{-1}_{Y^{1}_{i0}}(q)\right)\right)\qquad\theta_u(q)=\min\{\phi_u(q),\tilde{\phi}_u(q)\}.\]
	\end{proof}

	From the form of the intervals, one can find that the way to build bounds for the quantile treatment effect is different from what I have done in the linear outcome case. The key difference is that I need both of these two assumptions to build the identified set in the DID model, whereas in the change-in-changes setup, either restriction on anticipation probability or magnitude anticipation is possible to provide an identified set. $\phi_u(q)$ and $\phi_l(q)$ are possible bounds based on the magnitude restrictions following the idea that to the maximum extent, the pre-treatment potential outcomes without anticipation and treatment deviate from the observed pre-treatment outcomes up to a magnitude that is equal to the treatment effect. If one can find a solution to the formula, the solutions will be the corresponding bounds for the treatment effect. $\tilde{\phi}_u(q)$ and $\tilde{\phi}_l(q)$ are bounds obtained from the restrictions on the anticipation probability. The intuition for the bounds is that if one is interested in the performance of the $q$-th quantile treatment effect and the proportion of those who anticipate is up to $\pi$, then at least $(q-\pi)$ of the observation is not affected by anticipation. For example, if the anticipation effect is positive, the $q$-th quantile of potential outcomes without anticipation and treatment should be no less than the $(q-\pi)$-th quantile of observed outcomes, because at least these parts of the units do not anticipate, and no distortion exists. Similar logic can be used to analyze other cases. Although relying on fewer assumptions seems to represent an advantage over the DID model, and having two approaches can provide a tighter bound if researchers are willing to impose both assumptions, this approach has its own problem. Neither the solution to the formula nor the effectiveness of the bound achieved from the probability restriction are guaranteed. The former depends on the specific distribution properties of the potential outcomes and support of random variables, whereas the latter depends on the relationship between $q$ one is interested in and $\pi$ one chooses. For example, if $q\le\pi$, a negative quantile does not provide any useful information and this bound is useless. For the purpose of covering more cases, I suggest imposing two assumptions and choosing the tighter one as the identified set while applying this method.	
\end{proof}

\begin{proof}[Empirical Related Tables]
	This section provides extra empirical results for the empirical application using data from subject-specific teachers.
	\begin{table}[htbp]	
		\centering
		\begin{threeparttable}
			\caption{Effects of the Early Retirement Incentive Program on Scores}
			
			\begin{tabular}{ccccc}
				\toprule
				\toprule
				& \multicolumn{2}{c}{Original Results} & \multicolumn{2}{c}{With Anticipation} \\
				\cmidrule{2-5}          & Math  & Reading & Math  & Reading \\
				\midrule
				All Grade & 0.013 & 0.013 & [0.008,0.013] & [0.008,0.013] \\
				& (0.008) & (0.007) &       &  \\
				& [-0.002,0.028] & [-0.001,0.027] & [-0.005,0.026] & [-0.004,0.025] \\
				Grade 3 & 0.01  & -0.003 & [0.007,0.010] & [-0.007,-0.003] \\
				& (0.013) & (0.01)  &       &  \\
				& [-0.015,0.035] & [-0.024,0.017] & [-0.016,0.033] & [-0.05,0.04] \\
				Grade 6 & 0.004 & 0.016 & [0.002,0.004] & [0.010,0.016] \\
				& (0.01)  & (0.009) &       &  \\
				& [-0.016,0.024] & [-0.002,0.035] & [-0.017,0.023] & [-0.006,0.033] \\
				Grade 8 & 0.032 & 0.03  & [0.020,0.032] & [0.019,0.03] \\
				& (0.018) & (0.017) &       &  \\
				& [-0.004,0.067] & [-0.003,0.063] & [-0.011,0.063] & [-0.01,0.059] \\
				\bottomrule
			\end{tabular}%
			
			\label{subject}%
			\begin{tablenotes}
				\small 
				\item \textbf{Notes:} This table contains results using data from subject specific teachers. Each column presents results from a separate regression. Teachers who teach multiple grades are included in each grade. Teachers who teach in self-contained classrooms are assumed to teach both math and English. I list identified sets in the first row and 95\% level confidence sets in the third row for each result with anticipation. For comparison purposes, I also provide estimators, standard errors, and 95\% confidence intervals for results from \cite{fitzpatrick2014early}. Standard errors are displayed with parentheses.
			\end{tablenotes}
		\end{threeparttable}
	\end{table}%
\end{proof}

\begin{proof}[Estimation and Inference]
	
	Here I calculate these estimators for inference explicitly. To be consistent with the setup in the empirical application, I work on the case in which $\tau_g\le0\le\mu_g$ as an example, and the other cases can be analyzed similarly. With these assumptions, I know $\mu_g\in[\mu_{g,l},\mu_{g,u}]$, where $\mu_{g,l}=\frac{m_g}{1+\pi}$ while $\theta_{u}=m_g$.
	Corresponding estimators for the lower and upper bounds of the interval will be
	\begin{eqnarray*}
		\hat{\mu}_{g,u}&=&\frac{1}{n_1}\sum_{i=1}^{n}[g(Y_{i1})-g(Y_{i0})]D_i-\frac{1}{n_0}\sum_{i=1}^{n}[g(Y_{i1})-g(Y_{i0})](1-D_i)\qquad\\ 
		\hat{\mu}_{g,l}&=&\frac{\hat{\mu}_{g,u}}{1+\hat{\pi}}\qquad n_1=\sum_{i=1}^{n}D_i\qquad n_0=n-n_1\qquad\hat{\pi}\ \text{is a consistent estimator for $\pi$.}
	\end{eqnarray*}
	For simplicity, I define $\hat{\delta}_{dt}=\frac{1}{n_d}\sum_{k=1}^{n}g(Y_{kt})\mathbb{I}[D_k=d]$ for $(d,t)\in\{0,1\}^{2}$. I use $d$ to index the group and $t$ for time period. I can also define \[\hat{\sigma}_{dt}^{2}=\frac{1}{n_d-1}\sum_{k=1}^{n}[g(Y_{kt})-\hat{\delta}_{dt}]^{2}\mathbb{I}[D_k=d]\]
	\[\hat{\text{cov}}_{d}=\frac{1}{n_d-1}\sum_{k=1}^{n}[g(Y_{k1})-\hat{\delta}_{d1}][g(Y_{k0})-\hat{\delta}_{d0}]\mathbb{I}[D_k=d],\] which measures the variance for group d at period t and covariance between two time periods for group d. For well-behaved function $g(.)$, which guarantees the asymptotic normality of average, I have
	\[\hat{\sigma}_{u}^{2}=\frac{\hat{\sigma}_{11}^{2}+\hat{\sigma}_{10}^{2}-2\hat{\text{cov}}_{1}}{\hat{p}}+\frac{\hat{\sigma}_{01}^{2}+\hat{\sigma}_{00}^{2}-2\hat{\text{cov}}_{0}}{1-\hat{p}}\]
	\[\hat{\sigma}_{l}^{2}=\frac{\hat{\sigma}_{11}^{2}+\hat{\sigma}_{10}^{2}-2\hat{\text{cov}}_{1}}{\hat{p}(1+\hat{\pi})^{2}}+\frac{\hat{\sigma}_{01}^{2}+\hat{\sigma}_{00}^{2}-2\hat{\text{cov}}_{0}}{(1-\hat{p})(1+\hat{\pi})^{2}}\]
	with $\hat{p}=\frac{n_1}{n}$.
	
	For the validity of the inference procedure, when Assumptions (i) and (ii) are satisfied and the upper- and lower-bound estimators are ordered by construction, the procedure of \cite{imbens2004confidence} is valid according to \cite{stoye2009more}. All that remains to be addressed is that because the variances of the lower and upper bounds might change so to guarantee the effectiveness of the confidence set, one needs to choose the larger variance for both bounds.\\
	\begin{proof}[Proof of Theorem \ref{infer}]
		
		To keep consistency with the estimation procedure, I focus on the case
		\[\mu_g\in\left[\min\left\{m_g,\frac{m_g}{1+\pi}\right\},\max\left\{m_g,\frac{m_g}{1+\pi}\right\}\right]\qquad\hat{\mu}_{g,u}=\hat{m}_g\qquad\hat{\mu}_{g,l}=\frac{\hat{\mu}_{g,u}}{1+\hat{\pi}}.\]
		If $m_g>0$ and thus $\mu_g>0$, then based on Assumptions (i) and (ii) plus the ordered upper- and lower-bound estimators, the proposed interval is effective following \cite{imbens2004confidence} and \cite{stoye2009more}. However, if $m_g<0$ and thus $\mu_g<0$ but I get $\hat{m}_g>0$, then although $\hat{\mu}_{g,u}$ is still larger, now it is the estimator for the lower bound. For notation simplicity, use $\lambda$ to represent $\mathbb{P}[A_i=1|D_i=1]$, and I know $\lambda\in[0,\pi]$. Define $\sigma=\max\{\sigma_l,\sigma_u\}$, where $\sigma$ is the standard deviation of $m_g$:
		\begin{eqnarray*}
			&\mathbb{P}&\left(\frac{\hat{m}_g}{1+\hat{\pi}}-C_{n}\frac{\hat{\sigma}}{\sqrt{n}}\le\frac{m_g}{1+\lambda}\le\hat{m}_g+C_{n}\frac{\hat{\sigma}}{\sqrt{n}}\right)\\
			=&\mathbb{P}&\left(\frac{\sqrt{n}\frac{\hat{m}_g(\lambda-\hat{\pi})}{1+\hat{\pi}}-C_{n}\hat{\sigma}(1+\lambda)}{\sigma}\le\sqrt{n}\frac{m_g-\hat{m}_g}{\sigma}\le\frac{\sqrt{n}\lambda\hat{m}_g+C_{n}\hat{\sigma}(1+\lambda)}{\sigma}\right).
		\end{eqnarray*}
		For $\varepsilon>0$, $N_0$ exists such that $N>N_0$, I have $\lvert\frac{\hat{\sigma}-\sigma}{\sigma}\rvert<\varepsilon$, and thus, $\varepsilon>1-\frac{\hat{\sigma}}{\sigma}$. The probability then satisfies
		\begin{eqnarray*}
			&\ge&\mathbb{P}\left(\frac{\sqrt{n}\frac{\hat{m}_g(\lambda-\hat{\pi})}{1+\hat{\pi}}-C_{n}\sigma(1+\lambda)(1-\varepsilon)}{\sigma}\le\sqrt{n}\frac{m_g-\hat{m}_g}{\sigma}\le\frac{\sqrt{n}\lambda\hat{m}_g+C_{n}\sigma(1+\lambda)(1-\varepsilon)}{\sigma}\right).
		\end{eqnarray*}
		Use $\Phi$ to represent the cumulative value of the standard normal distribution and $\phi$ to represent the p.d.f for the standard normal here. By Berry-Essen's central limit theorem, this term is arbitrarily close to 
		\begin{eqnarray*}
			&\Phi&\left(\frac{\sqrt{n}\lambda\hat{m}_g}{\sigma}+C_{n}(1+\lambda)(1-\varepsilon)\right)-\Phi\left(\frac{\sqrt{n}\frac{\hat{m}_g(\lambda-\hat{\pi})}{1+\hat{\pi}}}{\sigma}-C_{n}(1+\lambda)(1-\varepsilon)\right)\\
			=&\Phi&\left(\frac{\sqrt{n}\lambda\hat{m}_g}{\sigma}+C_{n}(1+\lambda)\right)-\Phi\left(\frac{\sqrt{n}\frac{\hat{m}_g(\lambda-\hat{\pi})}{1+\hat{\pi}}}{\sigma}-C_{n}(1+\lambda)\right)+2(1+\lambda)\varepsilon C_n\phi(\omega)
		\end{eqnarray*}
		for some $\omega$. $C_n$ is bounded and $\varepsilon$ can be arbitrarily small, so the last term can be ignored. Using similar logic, the left term can be written as
		\begin{eqnarray*}
			&\ge&\Phi\left(\frac{\sqrt{n}\lambda\hat{m}_g}{\hat{\sigma}}+C_{n}(1+\lambda)\right)-\Phi\left(-\frac{\sqrt{n}\frac{\hat{m}_g(\hat{\pi}-\lambda)}{1+\hat{\pi}}}{\hat{\sigma}}-C_{n}(1+\lambda)\right)\\
			&-&C_0\varepsilon\lambda\left(\sqrt{n}\frac{\hat{m}_g-m_g}{\hat{\sigma}}+\sqrt{n}\frac{m_g}{\hat{\sigma}}\right)\phi(\omega^{\prime})
		\end{eqnarray*}
		for some other $\omega^{\prime}$ and constant number $C_0$. The first term within parentheses is normally distributed, the second term is negative, and I can take $\varepsilon$ arbitrarily small. Recall that $\lambda\in[0,\pi]$ and the smallest value is taken at $\lambda=0$. Thus,
		\begin{eqnarray*}
			&\mathbb{P}&\left(\frac{\hat{m}_g}{1+\hat{\pi}}-C_{n}\frac{\hat{\sigma}}{\sqrt{n}}\le\frac{m_g}{1+\lambda}\le\hat{m}_g+C_{n}\frac{\hat{\sigma}}{\sqrt{n}}\right)\\
			\ge&\Phi&\left(C_{n}\right)-\Phi\left(-\frac{\sqrt{n}(\hat{\mu}_{g,u}-\hat{\mu}_{g,l})}{\hat{\sigma}}-C_{n}\right)=\alpha.
		\end{eqnarray*}
	\end{proof}

	\begin{proof}[Proof of Corollary \ref{null0}]
		
		Take the case $\tau_g\le0\le\mu_g$ for example. The new confidence set has the form
		\[\left[\frac{\hat{\mu}_{g,u}}{1+\hat{\pi}}-C_n\frac{\hat{\sigma}}{\sqrt{n}},\hat{\mu}_{g,u}+C_n\frac{\hat{\sigma}}{\sqrt{n}}\right]\]
		with 
		\[\Phi\left(C_n+\sqrt{n}\frac{\hat{\mu}_{g,u}-\hat{\mu}_{g,l}}{\hat{\sigma}}\right)-\Phi(-C_n)=\Phi\left(C_n+\frac{\hat{\pi}}{1+\hat{\pi}}\tilde{t}\right)-\Phi(-C_n)=\alpha.\]
		Given that the right hand side of the confidence set is always positive, I only need to compare $\frac{\hat{\mu}_{g,u}}{1+\hat{\pi}}-C_n\frac{\hat{\sigma}}{\sqrt{n}}$ and 0:
		\[\frac{\hat{\mu}_{g,u}}{1+\hat{\pi}}-C_n\frac{\hat{\sigma}}{\sqrt{n}}=\frac{\hat{\sigma}}{(1+\hat{\pi})\sqrt{n}}(\tilde{t}-C_n(1+\hat{\pi})).\]
		I know $\Phi\left(C_n+\frac{\hat{\pi}}{1+\hat{\pi}}\tilde{t}\right)-\Phi(-C_n)$ is increasing in both $C_n$ and $\tilde{t}$, so the solution for $C_n$ is a decreasing function for $\tilde{t}$ at any given $\pi$. For any specific $\pi$, the smallest $\tilde{t}$ that guarantees $\frac{\tilde{t}}{1+\pi}\ge C_n$ is the value that solves
		\[\Phi\left(\tilde{t}\right)-\Phi\left(-\frac{\tilde{t}}{1+\pi}\right)=\alpha.\]
		The left hand side is an increasing function in $\tilde{t}$ and a decreasing function in $\pi$, so the worst scenario happens at $\pi=1$, giving the expression $\Phi(t^{*})-\Phi(-t^{*}/2)=\alpha$. As long as $\tilde{t}>t^{*}$, $\frac{\tilde{t}}{1+\pi}>C_n$ is guaranteed and I can conclude $0\not\in CS_{\alpha}^{\mu}$.
	\end{proof}

\end{proof}

\begin{proof}[Alternative Assumptions and Corresponding Bounds]
	
	Sometimes, proposing specific assumptions other than what have already been assumed may be reasonable. In this part, I discuss several alternative assumptions that will also give partial identification results under different circumstances.
	
	The first alternative approach mentioned here identifies the parameter of interest through a boundary on the outcomes. The advantage is that this assumption satisfies naturally under certain setups, for example, binary outcomes and setups when the outcomes are scores. Using this approach does not need to make discussions based on different signs of the treatment effects and does not need to bound the magnitude of effect. The approach will also benefit from a certain choice of bounded $g(.)$ function. However, I may need some other assumptions to validate it.
	\begin{saassumption}
		\label{homo}
		$\mathbb{E}[g(Y_{i0}(0,0))|D_i=1,A_i=0]=\mathbb{E}[g(Y_{i0}(0,0))|D_i=1,A_i=1]$.
	\end{saassumption}
	This assumption states that for those who will get treated, their potential outcome without treatment in the first period should be the same across those who anticipate and those who do not. This assumption is weaker than an ``independent anticipation'' assumption, and I only impose this restriction for the treated group.
	\begin{saassumption}
		\label{outcomebound}
		The outcome variable in the first period is bounded, which means $g(Y_{i0})\in[a,b]$.
	\end{saassumption}
	
	The bounded outcome assumption is not quite restrictive and fits into different situations, for example, the binary outcome case. It will naturally have a bounded outcome between 0 and 1. For example, in other cases, when the outcome is score between 0 and 100 or the $g(.)$ function itself is bounded, this assumption is automatically satisfied.
	\begin{satheorem}
		Under Assumptions \ref{sampling}-\ref{paralleltrends}, and  Assumptions \ref{homo} and \ref{outcomebound}, the parameter of interest $\mu_g$ is partially identified in a closed interval. The lower and upper bounds of that interval $\mu_{g,l}$, $\mu_{g,u}$ have the following form:
		\[\mu_{g,l}=\mathbb{E}\left[\frac{D_i-\mathbb{P}[D_i=1]}{\mathbb{P}[D_i=1](1-\mathbb{P}[D_i=1])}g(Y_{i1})+\frac{1-D_i}{1-\mathbb{P}[D_i=1]}g(Y_{i0})-b\right]\]
		\[\mu_{g,u}=\mathbb{E}\left[\frac{D_i-\mathbb{P}[D_i=1]}{\mathbb{P}[D_i=1](1-\mathbb{P}[D_i=1])}g(Y_{i1})+\frac{1-D_i}{1-\mathbb{P}[D_i=1]}g(Y_{i0})-a\right]\].
	\end{satheorem}
	The proof is relatively straightforward because I can see that for the unobserved term, I will have
	\[\mathbb{E}[g(Y_{i0}(0,0))|D_i=1]=\mathbb{E}[g(Y_{i0})|D_i=1,A_i=0]\]
	under Assumption \ref{homo}, which can help build a bound by conditional expectation. The case with covariates is similar after taking all the things conditionally and is skipped here.
	
	One may argue the identified set given by a bounded outcome might be too loose in some circumstances, because this approach just uses the upper and lower bound of the outcome itself. If people do not like to put restrictions on the signs and magnitudes of treatment effects, I can also provide a corresponding identification result under some other assumptions. The idea is that I may not observe the anticipation status for everybody; however, if I have a bound for the anticipation probability among the group I am interested in, I can get a bound about the treatment effect by assigning those with the highest or lowest outcomes to the anticipation treatment. Among all possible anticipation treatment statuses satisfying the proportion restriction, I choose the worst and best one based on the observed outcomes so I can get corresponding upper and lower bounds. 
	\begin{satheorem}
		Under Assumptions \ref{sampling}-\ref{paralleltrends}, \ref{prob_bound} and Assumption \ref{homo}, I can build a closed interval for the parameter of interest $\mu_g$ with upper and lower bounds $\mu_{g,l}$ and $\mu_{g,u}$ satisfying 
		\begin{eqnarray*}
		\mu_{g,l}&=&\mathbb{E}\left[\frac{D_i-\mathbb{P}[D_i=1]}{\mathbb{P}[D_i=1](1-\mathbb{P}[D_i=1])}g(Y_{i1})+\frac{1-D_i}{1-\mathbb{P}[D_i=1]}g(Y_{i0})\right]\\
		&-&\mathbb{E}[g(Y_{i0})|D_i=1,g(Y_{i0})\ge g(Y_{i0})_{\eta}]\\
		\mu_{g,u}&=&\mathbb{E}\left[\frac{D_i-\mathbb{P}[D_i=1]}{\mathbb{P}[D_i=1](1-\mathbb{P}[D_i=1])}g(Y_{i1})+\frac{1-D_i}{1-\mathbb{P}[D_i=1]}g(Y_{i0})\right]\\
		&-&\mathbb{E}[g(Y_{i0})|D_i=1,g(Y_{i0})\le g(Y_{i0})_{1-\eta}]
		\end{eqnarray*}
		$g(Y_{i0})_{\alpha}$ denotes the $\alpha$-th quantile of $g(Y_{i0})$ conditional on $D_i=1$ and $\eta=\pi$.
	\end{satheorem}
	The proof for this theorem follows the idea that I can assign the anticipation treatment group to those with the highest or lowest outcomes under a bounded proportion. This provide bounds for the purpose of partial identification. For these bounds, all the items are observable and can be estimated, and I do not need to put restrictions on signs and magnitudes of the treatment effect.
\end{proof}
\end{document}